\documentclass[12pt]{article}
\pdfoutput=1
\usepackage{jheppub}
\usepackage[skip=2pt,font=footnotesize]{caption}
\usepackage[utf8]{inputenc}
\usepackage{verbatim}
\usepackage{amsmath}
\usepackage{amssymb}
\usepackage{amsthm}
\usepackage{slashed}
\usepackage{amsfonts}
\usepackage{makecell}
\usepackage{tabularx}
\usepackage{comment}
\usepackage{subfigure}
\usepackage{float}
\usepackage{color}
\usepackage{physics}

\newtheorem{thm}{Theorem}[section]

\newtheorem{remark}{Remark}[section]

\newtheorem{exmp}{Example}[section]

\begin{document}
%\subheader{empty}

%\hbox{}

\title{Scrambling and decoding the charged quantum information}

\author[a,b]{Junyu Liu}

\affiliation[a]{Walter Burke Institute for Theoretical Physics, California Institute of Technology,\\ Pasadena, California 91125, USA}
\affiliation[b]{Institute for Quantum Information and Matter, California Institute of Technology,\\ Pasadena, California 91125, USA}

\abstract{Some deep conjectures about quantum gravity are closely related to the role of symmetries in the gravitational background, especially for quantum black holes. In this paper, we systematically study the theory of quantum information for a charged, chaotic system. We show how the quantum information in the whole system has been represented by its charge sectors, using the theory of quantum chaos and quantum error correction, with concrete examples in the context of the complex SYK model. We discuss possible implications for black hole thought experiments and conjectures about quantum gravity in the dynamical setup. We believe this work will have potential applications from theories of quantum gravity to quantum simulation in quantum devices.}
\maketitle

\section{Overview}
As C.N.Yang has stated, \emph{Symmetry dictates interaction}. Symmetry often plays a fundamental role in modern quantum physics. However, symmetry is pretty hard to understand when we discuss theory with semiclassical diffeomorphism invariance, namely, gravity, especially with the existence of black holes.

We will discuss two famous conjectures about symmetry in quantum gravity. Firstly, people believe that in a consistent quantum gravitational theory, there is no precise definition of global symmetry (see some early arguments, for instance, \cite{Misner:1957mt,Banks:2010zn}). In fact, since black holes cannot distinguish global symmetry according to the no-hair theorem, global symmetry in quantum gravity will lead to an infinite number of indistinguishable states and trouble for black hole remnants. This is called the \emph{no-global-symmetry conjecture}.

Secondly, it is conjectured that gravity is the weakest force for allowed quantum gravity theories. Roughly speaking, for quantum gravitational theory associated with U(1) gauge symmetry, there always exist states whose charge-to-mass ratios are larger than a universal lower bound, which is equal to $1/M_\text{planck}$. This is called the \emph{weak gravity conjecture}  \cite{ArkaniHamed:2006dz,Kats:2006xp,Cheung:2014vva,Harlow:2015lma,Cheung:2018cwt,Montero:2018fns,Cheung:2019cwi}. The original argument for this conjecture is also related to black holes: if all black hole states have small charge-to-mass ratios, those black holes are hard to decay, again causing a large number of states.  

Those two conjectures put significant constraints on the space of effective field theories that are allowed by rules of quantum gravity, sharpening our understanding about the boundary of the string theory landscape \cite{Susskind:1995da,Susskind:2003kw,Vafa:2005ui,ArkaniHamed:2006dz,Adams:2006sv,Ooguri:2006in}. Furthermore, symmetry might also be important towards a resolution of the black hole information paradox \cite{Hawking:1974sw}. For instance, There are proposals suggesting that supertranslation symmetry breaking may provide soft gravitons, and tracing out soft modes may provide a thermal spectrum during the Hawking radiation process \cite{Hawking:2016msc}.

Recently, important progress has been made about symmetry in quantum gravity. Combining technologies from holography and quantum information science, people formulate a precise notion of global and gauge symmetries in gravitational theory, and moreover, give a physical proof of the no-global-symmetry conjecture in holographic theories \cite{Harlow:2018jwu,Harlow:2018tng}. The proof is based on the quantum error correction theory of AdS/CFT: in AdS/CFT, the holographic dictionary is understood as an error correction code, where the code subspace corresponds to the low energy sector in CFT and effective field theory in the bulk \cite{Dong:2016eik,Almheiri:2014lwa,Harlow:2016vwg}. From the quantum information point of view, the no global symmetry statement is shown to be closely related to the Eastin-Knill theorem \cite{EK,EK2,EK3} in quantum error correction \cite{Hayden:2017jjm,Harlowtalk,Faist:2019ahr}: there is no exact covariant code associated with continuous global symmetry. 

The novel proof, given in \cite{Harlow:2018jwu,Harlow:2018tng}, shows the power of quantum information science when applying to contexts of quantum gravity. However, the proof itself cannot manifestly make use of black holes and their radiation, which are closely related to early intuitions about both conjectures. Thus, it is natural to ask the following question: is it possible to make statements about symmetry in quantum gravity, manifestly following early intuitions about the black hole radiation (decay) process?

In fact, there are toy models in quantum information science about black hole radiation. A simple model made by quantum circuits, proposed by Hayden and Preskill \cite {Hayden:2007cs} suggested that black hole scrambles information quickly during radiation, opening up the study of quantum chaos in the theory of quantum gravity. After realizing that black holes are fastest information scrambler in the universe \cite{Sekino:2008he,Lashkari:2011yi}, people use out-of-time-ordered correlators (OTOCs) \cite{Shenker:2013pqa,Shenker:2014cwa} to quantify the speed of scrambling. In fact, black holes are shown to be maximally chaotic and saturate the early-time chaos bound for OTOCs \cite{Maldacena:2015waa}. The Sachdev-Ye-Kitaev model (SYK) \cite{Sachdev:1992fk,Kitaev}, has been recently discussed as a simple model for quantum gravity, successfully reproducing maximal chaotic dynamics for black holes \cite{Maldacena:2016hyu}.

Those amazing developments motivate us to set up a toy model\footnote{Toy models are proven to be very useful in these years. For instance, quantum error correction \cite{Pastawski:2015qua}, quantum chaos \cite{Hayden:2007cs,Maldacena:2016hyu} and black hole information paradox \cite{Penington:2019npb,Almheiri:2019psf}.} of random unitary \cite{Hosur:2015ylk} with symmetry. We will consider the following simple model: consider a given Hamiltonian $H$ and a (global) charge operator $Q$, since 
\begin{align}
\left[ {H,Q} \right] = 0~,
\end{align}
we could decompose the whole Hilbert space into charge sectors, and they will separately evolve during time evolution. Namely, for a unitary operation $U=e^{iHt}$, we have\footnote{It differs from usual unitary time evolution by a sign in this convention, where we could simply redefine the energy eigenvalues with a minus sign to flip it.}
\begin{align}
&U = {e^{i( \oplus _{i = 1}^D{H_i})}}~,
\end{align}
where $H_i$s are Hamiltonian blocks for each charge sector, and the whole Hilbert space dimension is $L=2^D$ (for qubit systems). For simplicity, we discuss the U(1) symmetry, the simplest continuous symmetry in the quantum system. Then, the charge operator is defined as a sum of Pauli $Z$s
\begin{align}
Q = \sum\limits_{i = 1}^D {\frac{{1 + {Z_i}}}{2}}~,
\end{align}
and the U(1) action is given by a representation from the U(1) group, $e^{i \theta Q}$, where $q$ is a real number corresponding to the charge. 

The above settings are, of course, far from real black holes. The main problem with this setup is that it is not easy to define gauge symmetry and justify its difference from global symmetry. However, here we are going to make the simplest setup: since in the context of the Hayden-Preskill experiment, the unitary is completely random, there is no definition of the locality. Thus we may not need to distinguish global and gauge symmetries. This type of model is allowed by the Eastin-Knill theorem, since quantum error correction is approximate for those random codes. Moreover, when discussing gauge symmetry, we could set the experiment with global symmetry in the boundary, which is dual to gauge symmetry in the bulk. When discussing global symmetry, we could directly put the experiment in the bulk. A very similar construction is made in \cite{Yoshida:2018ybz}, but with a slightly different motivation. Usually, people use random circuits with U(1) symmetry to construct models with energy conservation \cite{C1,C2,Yoshida:2018ybz,thesis,Hunter-Jones:2018otn}, while here we are mostly discussing charges. We also hope that our work is not only helpful for high energy physics and black holes, but also for quantum information science itself, for understanding the role of charge conservation in quantum information processing.

This paper is organized in the following order.
\begin{itemize}
\item In Section \ref{PR}, we will give an introduction to some basic concepts of quantum information theory, especially the theory of quantum chaos and quantum error correction. We will also set up some notations used for this work. 
\item In Section \ref{CH}, we present a technical discussion on the theory of quantum chaos with U(1) symmetry, based on quantum information technology built mostly from \cite{Roberts:2016hpo}. We will present theorems and computations about spectral form factors, frame potentials, OTOCs, and decoupling properties. They are chaotic variables capturing scrambling properties of the system.  
\item In Section \ref{EG}, we give an explicit example, the complex SYK model, to support theories developed in Section \ref{CH} by numerical computations.
\item In Section \ref{CO}, we discuss quantum error correction theory for random unitaries with U(1) symmetry. Based on discussions about the Hayden-Preskill experiment, we present some arguments about the no-global-symmetry conjecture and the weak gravity conjecture. 
\item In Section \ref{DS}, we comment on the role of discrete symmetries. We discuss fundamental differences between discrete and continuous symmetry in terms of quantum chaos, quantum error correction, and quantum gravity. 
\item In Section \ref{CON}, we give an overview of potential research directions we are interested in related to this work. We mention further research in quantum gravity, conformal bootstrap, and quantum simulation/experiments in the real platform. 
\item In Appendices, we present a simple introduction of \texttt{Mathematica} package \texttt{RTNI} we are using for technical results in this work, and the alternative representation of form factors.
\end{itemize}

Some discussions in this paper, including Section \ref{CH} and part of Section \ref{CO}, are written in the mathematical form with theorems and proofs. These mathematical results are mostly prepared in order to support some physical claims about quantum chaos and quantum error correction of charged systems. For physics-oriented readers, one might consider jumping out of technical details and direct going to the physical conclusions and concrete examples, for instance, examples in Section \ref{EG} and physical discussions in Section \ref{CO} and \ref{DS}.

\section{Preliminaries}\label{PR}
In this section, we will set up our notations and give some simple introductions to basic quantum information concepts for completeness. 
\\
\\
\textbf{Notations}: We consider a qubit system with dimensions $L=2^D$. When discussing the black hole thought experiment, we also use $n$ to denote the number of qubits and $d=2^n$ to denote the dimension of the Hilbert space.

We use $(q,i)$ to denote the matrix elements, where $q$ means the charge sector $q$, and $i$ means the indices in the charge sector $q$. We sometimes ignore $i$ and directly use $q$ as a shorthand notation, which means the matrix block $q$. For instance,
\begin{align}
A_{{p_2}}^{{p_1}}B_{{p_1}}^{{p_2}} \equiv A_{({p_2},j)}^{({p_1},i)}B_{({p_1},i)}^{({p_2},j)}~.
\end{align}
Sometimes we also directly write $A_q$ as a shorthand notation of $A^{q}_q$. 

We use $\mathcal{H}$ to represent the Haar ensemble. $\tilde{B}$ means $UBU^\dagger$, and $\tilde{B}_q$ means $U_q B_q U^\dagger_q$.

We define
\begin{align}
{d_q} = \left( \begin{array}{l}
D\\
q
\end{array} \right)~,~~~~~~d_q^{(n)} = \left( {\begin{array}{*{20}{l}}
n\\
q
\end{array}} \right)~.
\end{align}
In the main text, the expectation value means 
\begin{align}
\left\langle {A(U)} \right\rangle  = \frac{1}{L}\int {dU{\rm{Tr}}\left( {A(U)} \right)}~,
\end{align}
and we also use the notation
\begin{align}
{\left\langle {{A_q}} \right\rangle _q} = \frac{1}{{{d_q}}}\int {d{U_q}{\rm{Tr}}\left( {{A_q}({U_q})} \right)}~.
\end{align}
Sometimes we use the subscript $\left\langle {A} \right\rangle_\mathcal{E} $ to denote that we are studying the average over ensemble $\mathcal{E}$. Thus we might also write the above notation as ${\left\langle {{A_q}} \right\rangle _{\mathcal{E}_q}}$.
\\
\\
\textbf{Basics about the Haar unitary}: The Haar ensemble is defined as a uniform measure for the unitary group. Formally, a Haar ensemble $\mathcal{H}$ defines a measure such that for every possible function $f$ on the unitary group, we have 
\begin{align}
\int_{\cal H} {dUf(U)}  = \int_{\cal H} {dUf(UV)}= \int_{\cal H} {dUf(VU)}~.
\end{align}
Namely, the distribution is both left and right invariant, where the measure is normalized 
\begin{align}
\int_\mathcal{H} {dU}  = 1~.
\end{align}
Using the Haar randomness, we could compute the Haar integral
\begin{align}
&\int_\mathcal{H} {U_{{j_1}}^{{i_1}} \ldots U_{{j_p}}^{{i_p}}U_{{j_1}'}^{\dagger,{i_1}'} \ldots U_{{j_p}'}^{{\dagger,i_p}'}dU} \nonumber\\
&= \sum\limits_{\alpha ,\beta  \in {S_p}} {\delta _{{j_{\alpha (1)}}'}^{{i_1}} \ldots \delta _{{j_{\alpha (p)}}'}^{{i_p}}\delta _{{j_1}}^{{i_{\beta (1)}}'} \ldots \delta _{{j_p}}^{{i_{\beta (p)}}'}{\rm{Wg}}({\alpha ^{ - 1}}\beta )}~,
\end{align}
where $\alpha,\beta$ are elements of permutation group $S_p$ over $1,2,\cdots,p$. If the numbers of $U$ and $U^\dagger$ in the integrand are not equal, the result of the integral is zero. The function $\text{Wg}$ is called the (unitary) Weingarten function, which could be computed in the group theory. For instance, for $S_1$ we have
\begin{align}
{\rm{Wg}}(1) = \frac{1}{L}~,
\end{align}
while for $S_2$ and $L\ge 2$ we have
\begin{align}
&{\rm{Wg}}(1,1) = \frac{1}{{{L^2} - 1}}~,\nonumber\\
&{\rm{Wg}}(2) = \frac{{ - 1}}{{L({L^2} - 1)}}~.
\end{align}
So we could derive the two most widely used the Haar integral formulas
\begin{align}
&\int {dU} U_j^iU_l^{\dagger k} = \frac{1}{L}\delta _l^i\delta _j^k~,\nonumber\\
&\int {dU} U_j^iU_l^{k}U_n^{\dagger m}U_p^{\dagger o} = \frac{1}{{{L^2} - 1}}\left( {\delta _n^i\delta _p^k\delta _j^m\delta _l^o + \delta _p^i\delta _n^k\delta _l^m\delta _j^o} \right)\nonumber\\
&- \frac{1}{{L({L^2} - 1)}}\left( {\delta _n^i\delta _p^k\delta _l^m\delta _j^o + \delta _p^i\delta _n^k\delta _j^m\delta _l^o} \right)~.
\end{align}
For more detailed information, see \cite{Roberts:2016hpo}. We also give a brief introduction to the symbolic computation of the Haar integrals in Appendix \ref{gra}.
\\
\\
\textbf{U(1)-symmetric Haar unitary}: We consider a direct sum of the Haar ensembles over charge sectors. Namely, we define the ensemble 
\begin{align}
{ \oplus _p}{\mathcal{H}_p} = \left\{ {{ \oplus _p}{U_p}:{U_p} \in {\mathcal{H}_p}} \right\}~,
\end{align}
where each $\mathcal{H}_p$ is a Haar random ensemble with dimension $d_p$ and charge sectors are independent. Since it is a direct sum, when performing the Haar integral, we need to be careful about which charge sectors the indices are in. Here we give some examples. 

\begin{exmp}
When computing 
\begin{align}
\int {dU} U_{(q,j)}^{(q,i)}U_{(p,l)}^{\dag (p,k)}~,
\end{align}
we just need to discuss two cases. When $q=p$, namely charge sectors are equal, then, naively, we get the same formula we have before
\begin{align}
\int {dU} U_{(q,j)}^{(q,i)}U_{(q,l)}^{\dag (q,k)} = \frac{1}{{{d_q}}}\delta _l^i\delta _j^k~.
\end{align}
When $q\ne p$, the integral has been factorized by two independent Haar integrals in different charge sectors, where each of them is zero. So we get
\begin{align}
\int {dU} U_{(q,j)}^{(q,i)}U_{(p,l)}^{\dag (p,k)} = \frac{1}{{{d_q}}}{\delta _{qp}}\delta _l^i\delta _j^k~.
\end{align}
Similarly, we could compute the higher moments
\begin{align}
\int {dU} U_{({q_1},{j_1})}^{({q_1},{i_1})}U_{({q_2},{l_1})}^{\dag ({q_2},{k_1})}U_{({p_1},{j_2})}^{({p_1},{i_2})}U_{({p_2},{l_2})}^{\dag ({p_2},{k_2})}~.
\end{align}
There are the following non-vanishing situations. $q_1=q_2=p_1=p_2=q$,
\begin{align}
&\int {dU} U_{(q,{j_1})}^{(q,{i_1})}U_{(q,{l_1})}^{\dag (q,{k_1})}U_{(q,{j_2})}^{(q,{i_2})}U_{(p,{l_2})}^{\dag (p,{k_2})} = \frac{1}{{d_q^2 - 1}}\left( {\delta _{{l_1}}^{{i_1}}\delta _{{j_1}}^{{k_1}}\delta _{{l_2}}^{{i_2}}\delta _{{j_2}}^{{k_2}} + \delta _{{l_2}}^{{i_1}}\delta _{{j_2}}^{{k_1}}\delta _{{l_1}}^{{i_2}}\delta _{{j_1}}^{{k_2}}} \right)\nonumber\\
&- \frac{1}{{(d_q^2 - 1){d_q}}}\left( {\delta _{{l_1}}^{{i_1}}\delta _{{j_2}}^{{k_1}}\delta _{{l_2}}^{{i_2}}\delta _{{j_1}}^{{k_2}} + \delta _{{l_2}}^{{i_1}}\delta _{{j_1}}^{{k_1}}\delta _{{l_1}}^{{i_2}}\delta _{{j_2}}^{{k_2}}} \right)~,
\end{align}
for $d_q>1$ ($d_q=1$ is trivial, it is just 1); $q_1=q_2=q$ and $p_1=p_2=p$ but $q\ne p$,
\begin{align}
\int {dU} U_{(q,{j_1})}^{(q,{i_1})}U_{(q,{l_1})}^{\dag (q,{k_1})}U_{(p,{j_2})}^{(p,{i_2})}U_{(p,{l_2})}^{\dag (p,{k_2})} = \frac{1}{{{d_p}{d_q}}}\delta _{{l_1}}^{{i_1}}\delta _{{j_1}}^{{k_1}}\delta _{{l_2}}^{{i_2}}\delta _{{j_2}}^{{k_2}}~,
\end{align}
and $q_1=p_2=q$, $q_2=p_1=p$ but $q\ne p$,
\begin{align}
\int {dU} U_{(q,{j_1})}^{(q,{i_1})}U_{(q,{l_2})}^{\dag (q,{k_2})}U_{(p,{j_2})}^{(p,{i_2})}U_{(p,{l_1})}^{\dag (p,{k_1})} = \frac{1}{{{d_p}{d_q}}}\delta _{{l_2}}^{{i_1}}\delta _{{j_1}}^{{k_2}}\delta _{{l_1}}^{{i_2}}\delta _{{j_2}}^{{k_1}}~.
\end{align}
\end{exmp}
Similar techniques could be generalized to other cases we are interested in.
\\
\\
\textbf{Form factor and frame potential}: For a given random unitary ensemble $\mathcal{E}$, we introduce the following two quantities. (Spectral) form factor $R_{2k}^{{\mathcal{E}}}$ and frame potential $F_\mathcal{E}^{(k)}$. 

Spectral form factors are widely used in random matrix theory, which are defined as the Fourier transform of spectral data for a given Hamiltonian ensemble. It is defined by
\begin{align}
R_{2k}^{{\mathcal{E}}} = \int _\mathcal{E}{dU{{\left| {{\rm{Tr(}}U{\rm{)}}} \right|}^{2k}}}~.
\end{align}
Since it is only related to the trace, it only cares about the eigenvalue distribution of the system. Thus, if we diagonalize the unitary operator as\footnote{Here in this notation, there might be ambiguities for defining $\lambda$ by a phase shift $2\pi$. However, we could imagine that all unitary ensembles we talk about here are generated by some Hamiltonian ensembles, and here the diagonalization means that we are diagonalizing the Hamiltonian.} 
\begin{align}
U = {\rm{diag}}\left( {{e^{i{\lambda _a}}}} \right)~,
\end{align}
and the eigenvalue measure of the given ensemble is given by $D\lambda $. Then the spectral form factor is given by
\begin{align}
R_{2k}^{\cal E} = \sum\limits_{a,b} {\int {D\lambda {e^{i({\lambda _{{a_1}}} +  \ldots  + {\lambda _{{a_k}}} - {\lambda _{{b_1}}} -  \ldots  - {\lambda _{{b_k}}})}}} }~.
\end{align}
In the discussions later, we will switch the eigenvalue basis and usual matrix basis freely. 

The spectral form factor could successfully capture the spectrum distribution in the Fourier space. For instance, for $k=1$ the form factor is just a Fourier transform of the spectrum distribution $\rho(\lambda)$,
\begin{align}
R_{2k}^{\cal E} = \sum\limits_a {\int {D\lambda {e^{i({\lambda _a})}}} }~.
\end{align}

For the Haar randomness, the spectral form factor is given by the following theorem \cite{sub,rain},
\begin{thm}
The $2k$ point form factor $R_{2k}^{\cal H}(L) $ counts for the number of permutations of $\{1,2,\cdots,k\}$ whose longest increasing subsequences are smaller or equal to $L$, where for a given permutation $\pi$, the increasing subsequence means that $i_1<i_2<i_3<\cdots$ such that $\pi(i_1)<\pi(i_2)<\pi(i_3)<\cdots$.
\end{thm}
Then we immediately know that 
\begin{thm}
For $k\le L$, 
\begin{align}
R_{2k}^{{\mathcal{H}}}(L) = \int {dU{{\left| {{\rm{Tr(}}U{\rm{)}}} \right|}^{2k}}} =k!~.
\end{align}
\end{thm}
Some alternative expressions are summarized in \cite{rain}. 

Now we introduce frame potential. Frame potential characterizes the 2-norm distance between a given ensemble and the Haar random unitary. It is defined as
\begin{align}
F_\mathcal{E}^{(k)} = \int_\mathcal{E} {dUdV\left| {{\rm{Tr}}(U{V^\dag })} \right|^{2k}}~.
\end{align}
We have the following simple observations 
\begin{thm}
\begin{align}
F_\mathcal{E}^{(k)} \ge F_\mathcal{H}^{(k)}~.
\end{align}
\end{thm}
\begin{proof}
Define 
\begin{align}
S = \int_{\cal E} {dU{U^{ \otimes k}} \otimes {{({U^\dag })}^{ \otimes k}}}  - \int_{\cal H} {dU{U^{ \otimes k}} \otimes {{({U^\dag })}^{ \otimes k}}}~.
\end{align}
We have
\begin{align}
&0 \le {\rm{Tr}}\left( {{S^\dag }S} \right) = \int_{\cal E} {dU} \int_{\cal E} {dV} {\left| {{\rm{Tr}}(U{V^\dag })} \right|^{2k}}\nonumber\\
&- 2\int_{\cal E} {dU} \int_{\cal H} {dV} {\left| {{\rm{Tr}}(U{V^\dag })} \right|^{2k}} + \int_{\cal H} {dU} \int_{\cal H} {dV} {\left| {{\rm{Tr}}(U{V^\dag })} \right|^{2k}}\nonumber\\
&= F_{\cal E}^{(k)} - 2F_{\cal H}^{(k)} + F_{\cal H}^{(k)} = F_{\cal E}^{(k)} - F_{\cal H}^{(k)}~,
\end{align}
where we have used the property of the Haar invariance.
\end{proof}
Moreover, for the Haar system, by the Haar invariance, we could simply observe that
\begin{thm}
\begin{align}
F_{\mathcal{H}}^{(k)} = R_{2k}^{\mathcal{H}}~.
\end{align}
\end{thm}
For further knowledge about form factor and frame potential, see \cite{Roberts:2016hpo,Cotler:2017jue}. 
\\
\\
\textbf{$k$-invariance}: $k$-invariance, introduced in \cite{Cotler:2017jue}, is a quantity that characterizes how \emph{invariant} it is under the Haar random unitary. For a given ensemble, $k$-invariance $I_\mathcal{E}^{(k)}$ is defined by
\begin{align}
I_{\cal E}^{(k)} = F_{\cal E}^{(k)} - F_{\tilde {\cal E}}^{(k)}~,
\end{align}
where $\mathcal{E}$ is from averaging ensemble $\mathcal{E}$ over the Haar measure,
\begin{align}
\tilde {\cal E} = \left\{ {\int_{\cal H} {dW} \left( {WU{W^\dag }} \right):U \in {\cal E}} \right\}~.
\end{align}

We know the following properties.
\begin{thm}
$k$-invariance is non-negative:
\begin{align}
I_{\cal E}^{(k)} \ge 0~.
\end{align}
\end{thm}
\begin{proof}
Introduce 
\begin{align}
T = \int_{\cal E} {dU{U^{ \otimes k}} \otimes {{({U^\dag })}^{ \otimes k}}}  - \int_{\tilde {\cal E}} {dU{U^{ \otimes k}} \otimes {{({U^\dag })}^{ \otimes k}}}~.
\end{align}
We have
\begin{align}
&0 \le {\rm{Tr}}({T^\dag }T) = \int_{\cal E} {dU} \int_{\cal E} {dV} {\left| {{\rm{Tr}}(U{V^\dag })} \right|^{2k}}\nonumber\\
&- \int_{\cal E} {dU} \int_{\cal E} {dV} \int_{\cal H} {dW} {\left| {{\rm{Tr}}({U^\dag }WV{W^\dag })} \right|^{2k}}\nonumber\\
&- \int_{\cal E} {dU} \int_{\cal E} {dV} \int_{\cal H} {dW} {\left| {{\rm{Tr}}(W{U^\dag }{W^\dag }V)} \right|^{2k}}\nonumber\\
&+ \int_{\cal E} {dU} \int_{\cal E} {dV} \int_{\cal H} {dW} \int_{\cal H} {dX} {\left| {{\rm{Tr}}(W{U^\dag }{W^\dag }XV{X^\dag })} \right|^{2k}}\nonumber\\
&= F_{\cal E}^{(k)} - F_{\tilde {\cal E}}^{(k)} = I_{\cal E}^{(k)}~.
\end{align}
\end{proof}
\begin{thm}
The Haar measure has zero $k$-invariance:
\begin{align}
I_{\cal H}^{(k)} = 0~.
\end{align}
\end{thm}
\begin{proof}
In fact we could prove a stronger statement than $I_{\cal H}^{(k)} = 0$\footnote{The difference between the statement $\int_{\mathcal{H}} {dU {\int_\mathcal{E} dVf(UV{U^\dag })} }  = \int _\mathcal{E}{dUf(U)} $ and $I_{\cal E}^{(k)} = 0$ is like the difference between left and right invariance and $F_{\cal \mathcal{E}}^{(k)}=k!$. Just like the fact that if $F_{\cal \mathcal{E}}^{(k)}=k!$, $\mathcal{E}$ is not necessarily left and right invariant, namely not necessarily Haar (it might be generically a $k$-design), $\int_{\mathcal{H}} {dU {\int_\mathcal{E} dVf(UV{U^\dag })} }  = \int _\mathcal{E}{dUf(U)} $ could imply $I_{\cal E}^{(k)} = 0$ but it is not easy to prove the reverse statement at least obviously.}. We show that for $\mathcal{E}=\mathcal{H}$ we have
\begin{align}
\int_{\mathcal{H}} {dU {\int_\mathcal{E} dVf(UV{U^\dag })} }  = \int _\mathcal{E}{dUf(U)}~.
\end{align}
In fact we define 
\begin{align}
g(V) = f(UV{U^\dag })~,
\end{align}
for given $f$. Then we have
\begin{align}
\int _\mathcal{H}{dVf(UV{U^\dag })}  = \int_\mathcal{H} {dVg(V)}  = \int_\mathcal{H} {dVg(VU)}  = \int_\mathcal{H} {dVf(UV)}~.
\end{align}
Thus
\begin{align}
\int_{\mathcal{H}^2} {dU {dVf(UV{U^\dag })} }  = \int_{\mathcal{H}^2}  {dU{dVf(UV)} }  = \int_{\mathcal{H}}  {dUf(U)}~.
\end{align}
\end{proof}
Thus, $k$-invariance $I_{\mathcal{E}}^{(k)}$ could measure the invariant property of ensemble $\mathcal{E}$ under the Haar average $\tilde{\mathcal{E}}$, which is similar to the fact that $F_{\mathcal{E}}^{(k)}$ could measure the difference between $\mathcal{E}$ and the Haar randomness $\mathcal{H}$. If an ensemble $\mathcal{E}$ satisfies $I_{\mathcal{E}}^{(k)}=0$, we say that the ensemble is $k$-invariant. A typical $k$-invariant system is the Gaussian Unitary Ensemble (GUE) \cite{Cotler:2017jue}.
\\
\\
\textbf{Decoupling} \cite{Page:1993df}: Consider a pure state $\rho_0=\ket{\psi_0}\bra{\psi_0}$, we make a Haar average $U\rho_0U^\dagger$. Assume that the whole system is decomposed as $A$ and $B$ where $d_A=\dim A$ and $d_B=\dim B$, we could define
\begin{align}
\Delta {\rho _A} = {\rho _A} - \frac{{{I_A}}}{{{d_A}}}~,
\end{align}
where
\begin{align}
{\rho _A} = {\rm{T}}{{\rm{r}}_B}U{\rho _0}{U^\dag }~,
\end{align}
and $I_A$ is the identity operator on $A$. From the Haar integral calculations, we know that
\begin{align}
\int_\mathcal{H} {dU\left\| {\Delta {\rho _A}} \right\|_1^2}  \le {d_A}\int_\mathcal{H} {dU\left\| {\Delta {\rho _A}} \right\|_2^2}  = \frac{{d_A^2 - 1}}{{{d_A}{d_B} + 1}} \approx \frac{{{d_A}}}{{{d_B}}}~,
\end{align}
where $\approx$ means we take the limit where $d=d_A d_B$ is large. This statement means that for the Haar random pure state when taking a small subsystem, we will obtain a nearly maximally mixed state. This is called the \emph{Page theorem}.
\\
\\
\textbf{Hayden-Preskill experiment}: The decoupling inequality (see for instance \cite{JP,Yoshida:2019kyp}) is famously used in the Hayden-Preskill experiment \cite{Hayden:2007cs}. We consider an initial product state in $A$ and $B$, where $A$ is sharing a Bell pair with $\bar{A}$, and $B$ is sharing a Bell pair with $\bar{B}$ (We assume that the dimension of the Hilbert space is the same for $A$ and $\bar{A}$, or $B$ and $\bar{B}$ respectively). Then we apply random unitary for the system $A$, $B$. After the unitary operation, see Figure \ref{fig1}, we find that the state $\rho_{\bar{A}C}$ is nearly decoupled to $\rho_{\bar{A}}$ and a maximally mixed state on $C$, namely,
\begin{align}
\int_\mathcal{H} {dU\left\| {\Delta {\rho _{\bar AC}}} \right\|_1^2}  \le \frac{{{d_A}{d_C}}}{{{d_B}{d_D}}}~,
\end{align}
where
\begin{align}
\Delta {\rho _{\bar AC}} = {\rho _{\bar AC}} - {\rho _{\bar A}} \otimes \frac{{{I_C}}}{{{d_C}}}~,
\end{align}
where we assume that $d_A\ll d_D$. One can show that decoupling between $\bar{A}$ and $C$ ensures that one can reconstruct the original state in $A$ with high fidelity. 
\begin{figure}[htbp]
  \centering
  \includegraphics[width=0.4\textwidth]{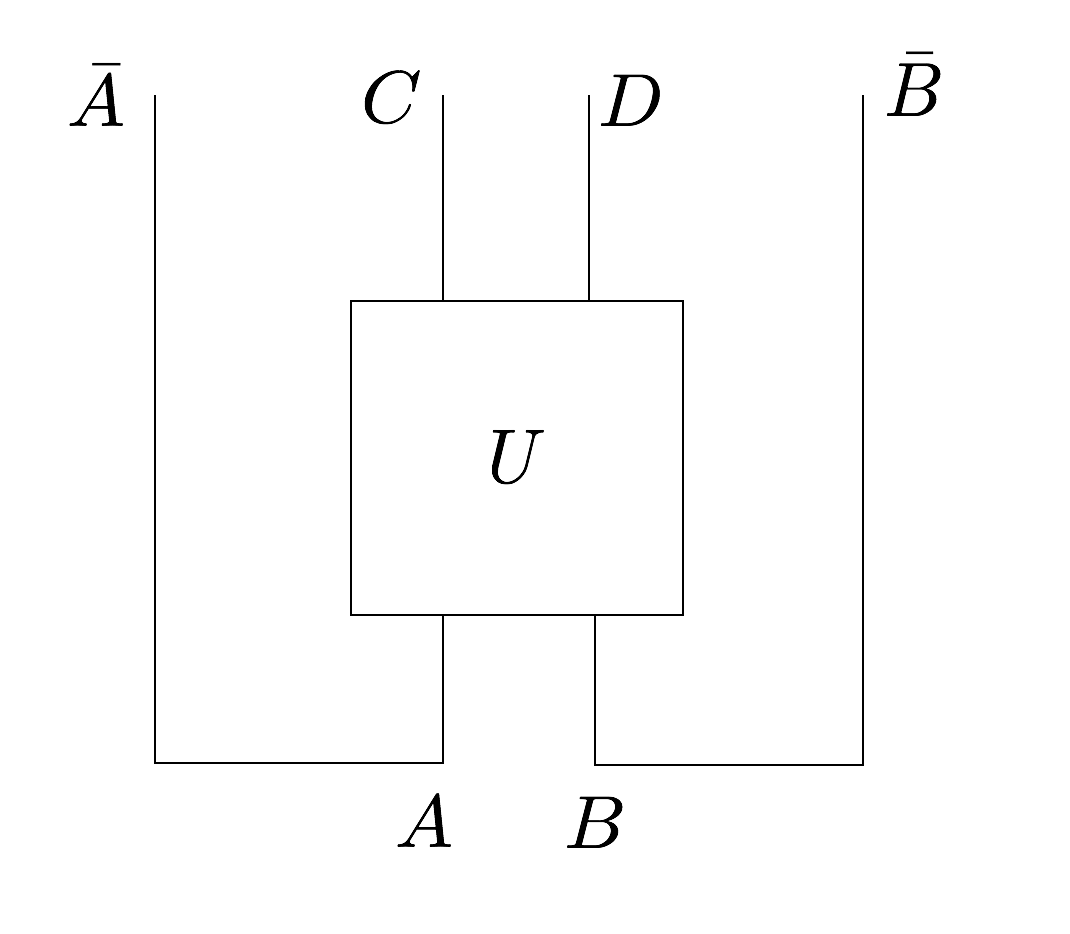}
  \caption{\label{fig1} Decoupling inequality/Hayden-Preskill experiment.}
\end{figure}
\\
\\
\textbf{Quantum error correction}: Quantum error correction is a well-established field in quantum information theory, focusing on studying how possible a quantum system could be protected against random errors. For a generic introduction, see \cite{JP}.  

We briefly review the basic ingredients that might be used in this paper. An error correction code is made by a Hilbert space $\mathcal{A}$, a noise channel (or noise combined with encoding) $\mathcal{N}:S(\mathcal{H}_A)\to S(\mathcal{H}_B)$, and a decoding map $\mathcal{D}:S(\mathcal{H}_B) \to S(\mathcal{H}_A)$. Here $S(\mathcal{H})$ means the space of density matrix on the Hilbert space $\mathcal{H}$. We usually have $\dim \mathcal{H}_A\le \dim \mathcal{H}_B$. The channel $\mathcal{N}$ could generically preserve the trace, 
\begin{align}
{\cal N}(\rho ) = \sum\limits_a {{{\cal N}_a}\rho {\cal N}_a^\dag }~,
\end{align}
with 
\begin{align}
\sum\limits_a {{\cal N}_a^\dag {{\cal N}_a}}  = 1~.
\end{align}
In this definition, $\mathcal{N}_a$s are Hermitian. This expansion is called Kraus representation, and $\mathcal{N}_a$s, called Kraus operators, are some subset of Pauli chains, specified by the error we consider. The basis vectors in $\mathcal{H}_A$ are called codewords, while $\mathcal{H}_A$ is called code subspace.

The requirement of error correction is that there exists $\mathcal{D}$ such that
\begin{align}
({\cal D} \circ {\cal N}) = {\rm{I}}{{}_A}~,
\end{align}
where $I_A$ means the identity operator on $\mathcal{H}_A$. Namely, the construction of the error correction code ensures that the information on code subspace $\mathcal{H}_A$ has been protected. 

A useful, necessary, and sufficient condition for quantum error correction is given by Knill and Laflamme \cite{KL}.
\begin{thm} 
Say that we have a code defined above. The necessary and sufficient condition for quantum error correction is given by
\begin{align}
\left\langle i \right|{\cal N}_a^\dag {{\cal N}_b}\left| j \right\rangle  = {C_{ab}}{\delta _{ij}}~,
\end{align}
where $\ket{i}$ and $\ket{j}$ are codewords, and $C_{ab}$ could be an arbitrary Hermitian matrix. This condition is required for every $a,b$ and $i,j$. 
\end{thm}
For a given code, one could use several parameters to justify the capability of the code against given errors, for instance, the dimension of the code subspace, the \emph{distance} of the code characterizing at most how large Paulis the code could correct, or \emph{fidelities} that are often used for approximate quantum error correction.

When designing a code, we often need to specify the noise. For instance, in the AdS/CFT code, the code subspace is the bulk effective field theory, and the encoding map is the AdS/CFT dictionary. The noise is specifically, erasing part of Hilbert space in the boundary (erasure). The AdS/CFT code is then protecting the bulk data from erasure errors, based on the claim of the entanglement wedge reconstruction. 
\\
\\
\textbf{Charge decoupling: }Generically, although symmetries make the Hamiltonian decomposed, it does not mean that energy eigenvalues will be completely independent in different charge sectors. 

One of the simplest examples might be $H=\lambda I$, where $I$ here is the identity operator, and $\lambda$ is a disordered parameter follow a given distribution. In this example, obviously, the operator $H$ commutes with every possible charge operator $Q$. So any operation could define a global charge for such a system. However, in each random realization, all eigenvalues are highly correlated. 

In this paper, we only consider models whose eigenvalues are not correlated in different charge sectors. We call it \emph{charge decoupling}, if the corresponding Hamiltonian satisfies this property, or at least roughly satisfies for large $L$. We find it roughly holds for the complex SYK model in the numerical simulation we have done. It will be interesting to study charged systems with highly correlated charge eigenspaces in the future.

\section{Chaos}\label{CH}
This is a technical section about computing chaotic variables in some charged random unitary ensembles. We will systematically compute spectral form factor, frame potential, OTOCs, and decoupling property using the theory of quantum chaos and the Haar integral. We focus on three different types of models: generic charged systems with charge decoupling, U(1)-symmetric Haar randomness, and $k$-invariant systems in each charge subspace. Some technical results are obtained using the \texttt{Mathematica} package, \texttt{RTNI}, where we give a simple introduction in Appendix \ref{gra}. Computations of chaotic variables in general or specific charged systems have their own values, while the discussion of decoupling will directly lead to some quantum error correction interpretations about chaotic systems, where we will give a more detailed discussion in Section \ref{CO}.
\subsection{Form factor}
\subsubsection{General result}
To start, we make some general assumptions about charged systems with charge decoupling. We consider the unitary is given by many charge sectors, and each charge sector acts independently on the state. In each charge sector, the unitary is generated by a chaotic Hamiltonian. Furthermore, we assume that each subspace $H_i$ has the eigenvalues $\lambda_{p,a}$, where $p,q$s are denoting charges, and $a$ is denoting the index of the eigenvalue inside the charge sector.

We will show the following theorem,
\begin{thm}
For U(1)-charged systems with charge decoupling, the $2k$-form factor could be represented by lower form factors in each charge subspace.
\end{thm}
\begin{proof}
We start by looking at lower point examples. For $R_2$, we have
\begin{align}
&R_2^{{ \oplus _p}{\mathcal{E}_p}}(L) = L + \sum\limits_{(p,a) \ne (q,b)} {\int {D\lambda {e^{i({\lambda _{p,a}} - {\lambda _{q,b}})}}} } \nonumber\\
&= L + \sum\limits_{p = q,a\ne b} {\int {D\lambda {e^{i({\lambda _{{a_p}}} - {\lambda _{{b_p}}})}}} }  + \sum\limits_{p \ne q} {\int {D\lambda {e^{i{\lambda _{{a_p}}}}}} \int {D\lambda {e^{ - i{\lambda _{{b_q}}}}}} }  \nonumber\\
&= \sum_{p } {R_2^{{\mathcal{E}_p}}({d_p})}  + \sum_{p \ne q} {R_1^{{\mathcal{E}_p}}\left( {{d_p}} \right)R_1^{{\mathcal{E}_q}*}\left( {{d_q}} \right)}~.
\end{align}
For future convenience, we could also define
\begin{align}
&R_2^{{{\cal E}_p}}({d_p}) = \sum\limits_{a,b} {\int {D\lambda } {e^{i({\lambda _{p,a}} - {\lambda _{p,b}})}}}  = \int {dU{\rm{Tr(}}U{\rm{)Tr(}}{U^\dag }{\rm{)}}} ~,\nonumber\\
&R_{21}^{{{\cal E}_p}}({d_p}) = \sum\limits_{a,b} {\int {D\lambda } {e^{i({\lambda _{p,a}} + {\lambda _{p,b}})}}}  = \int {dU{\rm{Tr(}}U{{\rm{)}}^2}} ~,\nonumber\\
&R_{22}^{{{\cal E}_p}}({d_p}) = \sum\limits_{a,b} {\int {D\lambda } {e^{i(2{\lambda _{p,a}} - {\lambda _{p,b}})}}}  = \int {dU{\rm{Tr(}}{U^2}{\rm{)Tr(}}{U^\dag }{\rm{)}}}~,\nonumber\\
&R_3^{{{\cal E}_p}}({d_p}) = \sum\limits_{a,b,c} {\int {D\lambda } {e^{i({\lambda _{p,a}} + {\lambda _{p,b}} - {\lambda _{p,c}})}}}  = \int {dU{\rm{Tr(}}U{\rm{)Tr(}}U{\rm{)Tr(}}{U^\dag }{\rm{)}}}~,\nonumber\\
&R_{31}^{{{\cal E}_p}}({d_p}) = \sum\limits_{a,b,c} {\int {D\lambda } {e^{i(2{\lambda _{p,a}} - {\lambda _{p,b}} - {\lambda _{p,c}})}}}  = \int {dU{\rm{Tr(}}{U^2}{\rm{)Tr(}}{U^\dag }{\rm{)Tr(}}{U^\dag }{\rm{)}}}~,\nonumber\\
&R_4^{{{\cal E}_p}}({d_p}) = \sum\limits_{a,b,c,d} {\int {D\lambda } {e^{i({\lambda _{p,a}} + {\lambda _{p,b}} - {\lambda _{p,c}} - {\lambda _{p,d}})}}}  = \int {dU{\rm{Tr(}}U{\rm{)Tr(}}U{\rm{)Tr(}}{U^\dag }{\rm{)Tr(}}{U^\dag }{\rm{)}}}~.
\end{align}
Furthermore, we could consider $R_4$,
\begin{align}
R_4^{{ \oplus _p}{\mathcal{E}_p}}(L) =\sum\limits_{\alpha,\beta,\gamma,\theta}{\int{D\lambda}{{e}^{i({{\lambda }_{\alpha}}+{{\lambda }_{\beta}}-{{\lambda }_{\gamma}}-{{\lambda }_{\theta}})}}}~.
\end{align}
We have 
\begin{align}
&\sum\limits_{\alpha ,\beta ,\gamma ,\theta } {\int {D\lambda } {e^{i({\lambda _\alpha } + {\lambda _\beta } - {\lambda _\gamma } - {\lambda _\theta })}}} \nonumber\\
&= \sum\limits_p {R_4^{{{\cal E}_p}}({d_p})}  + 4{\rm{Re}}\sum\limits_{p \ne q} {R_3^{{{\cal E}_p}}({d_p})R_1^{*{{\cal E}_q}}({d_q})}  + 4{\rm{Re}}\sum\limits_{p \ne q \ne u} {R_2^{{{\cal E}_p}}({d_p})R_1^{{{\cal E}_q}}({d_q})R_1^{*{{\cal E}_u}}({d_u})} \nonumber\\
&+ 2{\rm{Re}}\sum\limits_{p \ne q \ne u} {R_{21}^{{{\cal E}_p}}({d_p})R_1^{*{{\cal E}_q}}({d_q})R_1^{*{{\cal E}_u}}({d_u})}  + 2{\rm{Re}}\sum\limits_{p \ne q} {R_2^{{{\cal E}_p}}({d_p})R_2^{{{\cal E}_q}}({d_q})} \nonumber\\
&+ {\rm{Re}}\sum\limits_{p \ne q} {R_{21}^{{{\cal E}_p}}({d_p})R_{21}^{*{{\cal E}_q}}({d_q})}  + \sum\limits_{p \ne q \ne u \ne v} {R_1^{{{\cal E}_p}}({d_p})R_1^{{{\cal E}_q}}({d_q})R_1^{{{\cal E}_u}*}({d_u})R_1^{{{\cal E}_v}*}({d_v})}~.
\end{align}
We write an identical but alternative expression for this form factor in Appendix \ref{alter}.

Now we study the general case. Generically, for $2k$-point form factors, let $\sigma$ be a partition of $2k$ different objects, while term $i$ in the partition is specified as $\sigma_i$, and $i$ is ranging from 1 to $\ell(\sigma)$ (which contains $|\sigma_i|$ objects), the number of terms in the partition. Then the $2k$ point form factor is given by
\begin{align}
R_{2k}^{{ \oplus _p}{{\cal E}_p}}(L) = \sum\limits_{{\alpha _1},{\alpha _2}, \ldots ,{\alpha _k},{\beta _1},{\beta _2} \ldots ,{\beta _k}} {\int {D\lambda } {e^{i\sum\limits_{j = 1}^k {({\lambda _{{\alpha _j}}} - {\lambda _{{\beta _j}}})} }}}  = \sum\limits_\sigma  {\sum\limits_{{p_1} \ne {p_2} \ldots  \ne {p_{\ell (\sigma )}}} {\prod\nolimits_{i = 1}^{\ell (\sigma )} {{R^{\mathcal{E}_{p_i}}_{\left| {{\sigma _i}} \right|}}({\sigma _i},{d_{{p_i}}})} } }~,
\end{align}
where ${{R^{\mathcal{E}_{p_i}}_{\left| {{\sigma _i}} \right|}}({\sigma _i},{d_{{p_i}}})} $ means that we assign ${\alpha _1},{\alpha _2}, \ldots ,{\alpha _k},{\beta _1},{\beta _2} \ldots ,{\beta _k}$ to those $2k$ different objects, and compute form factors specified by $\sigma_i$, in the $p_i$ charge sector $\mathcal{E}_{p_i}$.
\end{proof}

\subsubsection{Haar randomness}
Now we specify the system to be U(1)-symmetric Haar $\mathcal{H}$. Firstly we could prove an asymptotic formula.
\begin{thm}
For large $D$ ($D\gg k\gg 1$), we have
\begin{align}
R_{2k}\sim k! D^k~.
\end{align}
\end{thm} 
\begin{proof}
In fact, in large $D$, $R_{2k}$ is dominated by 
\begin{align}
R_{2k}^{{ \oplus _p}{\mathcal{H}_p}}(L) \sim k!{\rm{Re}}\sum\limits_{{p_1} \ne {p_2} \ldots  \ne {p_k}} {R_2^{{\mathcal{H}}}({d_{{p_1}}})R_2^{{\mathcal{H}}}({d_{{p_2}}}) \ldots R_2^{{\mathcal{H}}}({d_{{p_k}}})}  \sim k!{D^k}~.
\end{align}
\end{proof}

\subsection{OTOCs}
\subsubsection{General result}
In general, the problem is highly simplified when operators commute with charge. 
\begin{thm}
For $A_i$, $B_i$ commuting with charge $Q$, we have
\begin{align}
{\left\langle {{A_1}{{\tilde B}_1}{A_2}{{\tilde B}_2} \ldots {A_k}{{\tilde B}_k}} \right\rangle _{\cal E}} = \sum\limits_p {\frac{{{d_p}}}{L}{{\left\langle {{A_{1,p}}{{\tilde B}_{1,p}}{A_{2,p}}{{\tilde B}_{2,p}} \ldots {A_{k,p}}{{\tilde B}_{k,p}}} \right\rangle }_{{{\cal E}_p}}}}~.
\end{align}
\end{thm}
\begin{proof}
We could firstly look at two-point examples. We have
\begin{align}
&\left\langle {A\tilde B} \right\rangle_\mathcal{E} = \frac{1}{L}\sum\limits_p {A_{(p,j)}^{(p,i)}B_{(p,l)}^{(p,k)}\int {dU} U_{(p,k)}^{\dag (p,j)}U_{(p,i)}^{(p,l)}} \nonumber\\
&= \frac{1}{L}\sum\limits_p {{d_p}\left\langle {{A_p}{{\tilde B}_p}} \right\rangle }_{\mathcal{E}_p}~.
\end{align}
Thus, in general, we have
\begin{align}
&\frac{1}{L}\int {dU{\rm{Tr}}\left( {{A_1}{U^\dag }{B_1}U{A_2}{U^\dag }{B_2}U \ldots {A_k}{U^\dag }{B_k}U} \right)}\nonumber\\
&= \frac{1}{L}\sum\limits_{p,q} \begin{array}{l}
A_{1,({p_2},{i_2})}^{({p_1},{i_1})}B_{1,({q_2},{j_2})}^{({q_1},{j_1})}A_{2,({p_4},{i_4})}^{({p_3},{i_3})}B_{2,({q_4},{j_4})}^{({q_3},{j_3})} \ldots A_{k,({p_{2k - 1}},{i_{2k - 1}})}^{({p_{2k - 1}},{i_{2k - 1}})}B_{k,({q_{2k - 1}},{j_{2k - 1}})}^{({q_{2k - 1}},{j_{2k - 1}})}\\
\int {dU} U_{({q_1},{j_1})}^{\dag ({p_2},{i_2})}U_{({p_3},{i_3})}^{({q_2},{j_2})}U_{({q_3},{j_3})}^{\dag ({p_4},{i_4})}U_{({p_5},{i_5})}^{({q_4},{j_4})} \ldots U_{({q_{2k - 1}},{j_{2k - 1}})}^{\dag ({p_{2k - 1}},{i_{2k - 1}})}U_{({p_1},{i_1})}^{({q_{2k - 1}},{j_{2k - 1}})}~.
\end{array} 
\end{align}
Since we know that $U$ is block-diagonal, so we have to force 
\begin{align}
&{p_2} = {q_1}~,\nonumber\\
&{p_3} = {q_2}~,\nonumber\\
&{p_4} = {q_3}~,\nonumber\\
&\ldots ~.
\end{align}
Since $A$ and $B$s are also block-diagonal, we have
\begin{align}
&{p_1} = {p_2}~,\nonumber\\
&{q_1} = {q_2}~,\nonumber\\
&{p_3} = {p_4}~,\nonumber\\
&\ldots ~.
\end{align}
So every index should be equal, and we have
\begin{align}
&{\left\langle {{A_1}{{\tilde B}_1}{A_2}{{\tilde B}_2} \ldots {A_k}{{\tilde B}_k}} \right\rangle _{\mathcal E}} \nonumber\\
&= \frac{1}{L}\sum\limits_p \begin{array}{l}
A_{1,(p,{i_2})}^{(p,{i_1})}B_{1,(p,{j_2})}^{(p,{j_1})}A_{2,(p,{i_4})}^{(p,{i_3})}B_{2,(p,{j_4})}^{(p,{j_3})} \ldots A_{k,(p,{i_{2k - 1}})}^{(p,{i_{2k - 1}})}B_{k,(p,{j_{2k - 1}})}^{(p,{j_{2k - 1}})}\\
\int {dU} U_{(p,{j_1})}^{\dag (p,{i_2})}U_{(p,{i_3})}^{(p,{j_2})}U_{(p,{j_3})}^{\dag (p,{i_4})}U_{(p,{i_5})}^{(p,{j_4})} \ldots U_{(p,{j_{2k - 1}})}^{\dag (p,{i_{2k - 1}})}U_{(p,{i_1})}^{(p,{j_{2k - 1}})}
\end{array} \nonumber\\
&= \sum\limits_p {\frac{{{d_p}}}{L}{{\left\langle {{A_{1,p}}{{\tilde B}_{1,p}}{A_{2,p}}{{\tilde B}_{2,p}} \ldots {A_{k,p}}{{\tilde B}_{k,p}}} \right\rangle }_{{{\cal E}_p}}}}~.
\end{align}
as desired.
\end{proof}
This theorem is simply expected since operators are also decoupled to different charge sectors. If we remove such assumptions, cases are a little harder. We will give the following simple example.
\begin{exmp}
We consider a generic two-point OTOC. We have
\begin{align}
&\left\langle {A\tilde B} \right\rangle  = \frac{1}{L}A_{(q,j)}^{(p,i)}B_{(p,l)}^{(q,k)}\int {dU} U_{(q,k)}^{\dag (q,j)}U_{(p,i)}^{(p,l)}\nonumber\\
&= \frac{1}{L}\sum\limits_p {A_{(p,j)}^{(p,i)}B_{(p,l)}^{(p,k)}\int {dU} U_{(p,k)}^{\dag (p,j)}U_{(p,i)}^{(p,l)}} \nonumber\\
&+ \frac{1}{L}\sum\limits_{p \ne q} {A_{(q,j)}^{(p,i)}B_{(p,l)}^{(q,k)}\int {dUU_{(q,k)}^{\dag (q,j)}} \int {dUU_{(p,i)}^{(p,l)}} }~.
\end{align}
The first term is the sum of all separate charge sectors. The second term is due to the non-vanishing of mixing blocks $A^{q}_p$ in matrix $A$ or $B$.
\end{exmp}
This mechanism is easy to obtain in the more general case, while similar to form factor calculation, we could compute partitions of $2k$ objects, and then assign each partition with known OTOCs in charge sectors. We will leave those exercises to curious readers.

\subsubsection{Haar randomness}
Now we consider the case where the ensemble is the U(1)-symmetric Haar randomness. To avoid triviality, we could consider cases where operators are randomly assigned instead of block-diagonal in charge eigenspaces. We proceed with this analysis by examples.
\begin{exmp}
We start from the two-point function. We have 
\begin{align}
{\left\langle {A\tilde B} \right\rangle _{{ \oplus _q}{\mathcal{H}_q}}}= \frac{1}{L}A_{({q_2},j)}^{({q_1},i)}B_{({q_4},l)}^{({q_3},k)}\int {dU} U_{({q_3},k)}^{\dag ({q_2},j)}U_{({q_1},i)}^{({q_4},l)}~.
\end{align}
Since $U$ is block-diagonal, we have to force $q_2=q_3$ and $q_1=q_4$ and to obtain a nontrivial Haar integral, we have to force every charge index to be equal. Thus we obtain 
\begin{align}
{\left\langle {A\tilde B} \right\rangle _{{ \oplus _p}{\mathcal{H}_p}}} = \sum\limits_p {\frac{{{d_p}}}{L}{{\left\langle {{A_p}{{\tilde B}_p}} \right\rangle }_{{\mathcal{H}_p}}}}~.
\end{align}
Using the Haar results we could obtain
\begin{align}
{\left\langle {A\tilde B} \right\rangle _{{ \oplus _p}{\mathcal{H}_p}}} = \sum\limits_p {\frac{{{d_p}}}{L}{{\left\langle {{A_p}} \right\rangle }_{{\mathcal{H}_p}}}{{\left\langle {{B_p}} \right\rangle }_{{\mathcal{H}_p}}}}~.
\end{align}
\end{exmp}
Specifically, we could consider $A$ and $B$ are Paulis. So we have the following example.
\begin{exmp}
Assuming $A$ and $B$ are Paulis. We know that the charge operator is generated by $Z$. For a given Pauli chain $\sigma$, we denote 
\begin{align}
&z(\sigma)=\#\text{ of }Z\text{s in the chain }\sigma~,\nonumber\\
&i(\sigma)=\#\text{ of }I\text{s in the chain }\sigma~.
\end{align}
Then we find 
\begin{align}
\left\langle {A\tilde B} \right\rangle_{{ \oplus _p}{\mathcal{H}_p}} = \left\{ \begin{array}{l}
\frac{1}{{{d_q}}}~~~z(A) = z(B) = q = D - i(A) - i(B)\\
0~~~{\rm{other cases}}
\end{array} \right.~.
\end{align}
\end{exmp}
Now we discuss higher-point functions.
\begin{exmp}
For the four-point function we have
\begin{align}
&{\left\langle {{A_1}{{\tilde B}_1}{A_2}{{\tilde B}_2}} \right\rangle _{{ \oplus _p}{\mathcal{H}_p}}}\nonumber\\
&=\frac{1}{L} A_{1,({q_2},j)}^{({q_1},i)}B_{1,({q_4},l)}^{({q_3},k)}A_{2,({q_6},n)}^{({q_5},m)}B_{2,({q_8},r)}^{({q_7},o)}\int {dU} U_{({q_3},k)}^{\dag ({q_2},j)}U_{({q_5},m)}^{({q_4},l)}U_{({q_7},o)}^{\dag ({q_6},n)}U_{({q_1},i)}^{({q_8},r)}\nonumber\\
&=\sum\limits_p {\frac{{{d_p}}}{L}{{\left\langle {{A_{1,p}}{{\tilde B}_{1,p}}{A_{2,p}}{{\tilde B}_{2,p}}} \right\rangle }_{{\mathcal{H}_p}}}} \nonumber\\
&+ \sum\limits_{p \ne q} {A_{1,(q,j)}^{(p,i)}A_{2,(p,i)}^{(q,j)}{{\left\langle {{B_1}} \right\rangle }_q}{{\left\langle {{B_2}} \right\rangle }_p}} {\rm{ }}\nonumber\\
&+ \sum\limits_{p \ne q} {{{\left\langle {{A_1}} \right\rangle }_q}{{\left\langle {{A_2}} \right\rangle }_p}B_{1,(p,j)}^{(q,i)}B_{2,(q,i)}^{(p,j)}}~,
\end{align}
where
\begin{align}
&{\left\langle {{A_{1,p}}{{\tilde B}_{1,p}}{A_{2,p}}{{\tilde B}_{2,p}}} \right\rangle _{{{\cal H}_p}}} = {\left\langle {{A_{1,p}}{A_{2,p}}} \right\rangle _{{{\cal H}_p}}}{\left\langle {{B_{1,p}}} \right\rangle _{{{\cal H}_p}}}{\left\langle {{B_{2,p}}} \right\rangle _{{{\cal H}_p}}} + {\left\langle {{A_{1,p}}} \right\rangle _{{{\cal H}_p}}}{\left\langle {{A_{2,p}}} \right\rangle _{{{\cal H}_p}}}{\left\langle {{B_{1,p}}{B_{2,p}}} \right\rangle _{{{\cal H}_p}}}\nonumber\\
&- {\left\langle {{A_{1,p}}} \right\rangle _{{{\cal H}_p}}}{\left\langle {{A_{2,p}}} \right\rangle _{{{\cal H}_p}}}{\left\langle {{B_{1,p}}} \right\rangle _{{{\cal H}_p}}}{\left\langle {{B_{2,p}}} \right\rangle _{{{\cal H}_p}}} - \frac{1}{{d_p^2 - 1}}{\left\langle {\left\langle {{A_{1,p}}{A_{2,p}}} \right\rangle } \right\rangle _{{{\cal H}_p}}}{\left\langle {\left\langle {{B_{1,p}}{B_{2,p}}} \right\rangle } \right\rangle _{{{\cal H}_p}}}~,
\end{align}
for $d_p> 2$. For $d_p=1$ we have 
\begin{align}
{\left\langle {{A_{1,p}}{{\tilde B}_{1,p}}{A_{2,p}}{{\tilde B}_{2,p}}} \right\rangle _{{{\cal H}_p}}} = {A_{1,p}}{B_{1,p}}{A_{2,p}}{B_{2,p}}~.
\end{align}
\end{exmp}

\subsubsection{$k$-invariant subspace}
Now, we start to study charged systems with $k$-invariance. What is the practical quantum information model with conserved charge, following the spirit of $k$-invariance for a general random unitary system? Considering a real system with charge decoupling, we expect that each charge sector should be completely independent of other sectors. Thus, it is natural to assign $k$-invariance \emph{in each charge subspace}. A practical model of this type is the complex SYK model: In each charge sector, the system looks like GUE, which is known to be $k$-invariant. Furthermore, the model almost has the property of charge decoupling generically in all time scale. Thus, we consider $k$-invariance in each subspace as a generalization of $k$-invariance in the case of U(1) symmetry. Note that $k$-invariance in each subspace may not imply $k$-invariance for the whole random unitary.

We start from the simplest case, where we assume that operators themselves are also independent in different charge sectors.
\begin{thm}
For operators commuting with the charge operator, we have 
\begin{align}
{\left\langle {{A_1}{{\tilde B}_1}{A_2}{{\tilde B}_2} \ldots {A_k}{{\tilde B}_k}} \right\rangle _\mathcal{E}} \approx \sum\limits_p {\frac{{R_{2k}^{{\mathcal{E}_p}}}}{{d_p^{2k}L}}{\rm{Tr}}\left( {{A_{1,p}}{B_{1,p}}{A_{2,p}}{B_{2,p}} \ldots {A_{k,p}}{B_{k,p}}} \right)}~.
\end{align}
\end{thm}
\begin{proof}
As proved before, we have
\begin{align}
{\left\langle {{A_1}{{\tilde B}_1}{A_2}{{\tilde B}_2} \ldots {A_k}{{\tilde B}_k}} \right\rangle _{\cal E}} = \sum\limits_p {\frac{{{d_p}}}{L}{{\left\langle {{A_{1,p}}{{\tilde B}_{1,p}}{A_{2,p}}{{\tilde B}_{2,p}} \ldots {A_{k,p}}{{\tilde B}_{k,p}}} \right\rangle }_{{{\cal E}_p}}}}~.
\end{align}
Now let us assume that in each sector, it highly deviates from the Haar results, and then the spectral form factors are large. In this case, we have
\begin{align}
{\left\langle {{A_{1,p}}{{\tilde B}_{1,p}}{A_{2,p}}{{\tilde B}_{2,p}} \ldots {A_{k,p}}{{\tilde B}_{k,p}}} \right\rangle _{{\mathcal{E}_p}}} \approx {\rm{Tr}}\left( {{A_{1,p}}{B_{1,p}}{A_{2,p}}{B_{2,p}} \ldots {A_{k,p}}{B_{k,p}}} \right)\frac{{R_{2k}^{{\mathcal{E}_p}}}}{{d_p^{2k + 1}}}~.
\end{align}
So the result is given by
\begin{align}
{\left\langle {{A_1}{{\tilde B}_1}{A_2}{{\tilde B}_2} \ldots {A_k}{{\tilde B}_k}} \right\rangle _\mathcal{E}} \approx \sum\limits_p {\frac{{R_{2k}^{{\mathcal{E}_p}}}}{{d_p^{2k}L}}{\rm{Tr}}\left( {{A_{1,p}}{B_{1,p}}{A_{2,p}}{B_{2,p}} \ldots {A_{k,p}}{B_{k,p}}} \right)}~.
\end{align}
\end{proof}
Now we give a two-point function example.
\begin{exmp}
We start from two-point. It is given by
\begin{align}
&\left\langle {A\tilde B} \right\rangle  = \frac{1}{L}A_{(q,j)}^{(p,i)}B_{(p,l)}^{(q,k)}\int {dU} U_{(q,k)}^{\dag (q,j)}U_{(p,i)}^{(p,l)}\nonumber\\
&= \sum\limits_p {\frac{{{d_p}}}{L}} {\left\langle {{A_p}{{\tilde B}_p}} \right\rangle _{{{\cal E}_p}}}\nonumber\\
&+ \frac{1}{L}\sum\limits_{p \ne q} {{\rm{Tr}}\left( {\int {d{U_q}d{V_q}A_q^p{V_q}{U^{\dag q}}V_q^\dag } \int {d{U_p}d{V_p}B_p^q{V_p}{U^p}V_p^\dag } } \right)}~.
\end{align}
The first term is
\begin{align}
{\left\langle {{A_p}{{\tilde B}_p}} \right\rangle _p} = {\left\langle {{A_p}} \right\rangle _p}{\left\langle {{B_p}} \right\rangle _p} + \frac{{R_2^{{{\cal E}_p}}({d_p}) - 1}}{{d_p^2 - 1}}{\left\langle {\left\langle {{A_p}{B_p}} \right\rangle } \right\rangle _p}~,
\end{align}
for $d_p>1$, and ${\left\langle {A_p\tilde{B}_p} \right\rangle _p} = A_pB_p$ for $d_p=1$. For the second term, we have
\begin{align}
&\frac{1}{L}\sum\limits_{p \ne q} {{\rm{Tr}}\left( {\int {d{U_q}d{V_q}A_q^p{V_q}{U^{\dag q}}V_q^\dag } \int {d{U_p}d{V_p}B_p^q{V_p}{U^p}V_p^\dag } } \right)} \nonumber\\
&= \frac{1}{L}\sum\limits_{p \ne q} {\frac{1}{{{d_p}{d_q}}}A_{(q,j)}^{(p,i)}B_{(p,i)}^{(q,j)}\int {d{U_p}d{U_q}} U_{(q,k)}^{\dag (q,k)}U_{(p,l)}^{(p,l)}} \nonumber\\
&= \frac{1}{L}\sum\limits_{p \ne q} {\frac{1}{{{d_p}{d_q}}}A_q^pB_p^qR_1^{{{\cal E}_p}}({d_p})R_1^{{{\cal E}_q}}({d_q})}~.
\end{align}
Specifically, if each charge sector is just the Haar system, the $R_1$ part and the $R_2$ part are zero, so we recover the previous result for the Haar randomness.  
\end{exmp}

This is only the two-point function. For higher-point functions, the computation is harder but straightforward based on the above methodology. 

\subsection{Frame potential}
\subsubsection{General result}
In general, frame potential is a much more complicated object. We have the following theorem. 
\begin{thm}
Frame potential for a general charged system $\mathcal{E}$ with charge decoupling could be written as a sum of variables inside charge sectors, although many of them cannot be represented as frame potentials in charge sectors.
\end{thm}
\begin{proof}
We consider the first frame potential to start, we have
\begin{align}
&F_{\mathcal E}^{(1)} = \int {dUdV\left( {\sum\limits_{p \ne q} {{\rm{Tr}}\left( {{U_p}V_p^\dag } \right){\rm{Tr}}\left( {{V_q}U_q^\dag } \right)}  + \sum\limits_p {{\rm{Tr}}\left( {{U_p}V_p^\dag } \right){\rm{Tr}}\left( {{V_p}U_p^\dag } \right)} } \right)} \nonumber\\
&= \int {dUdV\sum\limits_{p \ne q} {{\rm{Tr}}\left( {{U_p}V_p^\dag } \right){\rm{Tr}}\left( {{V_q}U_q^\dag } \right)} }  + \sum\limits_p {\int {dUdV} {\rm{Tr}}\left( {{U_p}V_p^\dag } \right){\rm{Tr}}\left( {{V_p}U_p^\dag } \right)} \nonumber\\
&= \sum\limits_{p \ne q} {\int {d{U_p}d{V_p}{\rm{Tr}}\left( {{U_p}V_p^\dag } \right)} \int {d{U_q}d{V_q}{\rm{Tr}}\left( {{V_q}U_q^\dag } \right)} }  + \sum\limits_p {F_{\mathcal{E}_p}^{(1)}}~.
\end{align}
In general,
\begin{align}
&F_{\mathcal E}^{(k)}= \int {dUdV{{\left( {\sum\limits_{p,q} {{\rm{Tr}}\left( {{U_p}V_p^\dag } \right){\rm{Tr}}\left( {{V_q}U_q^\dag } \right)} } \right)}^k}}\nonumber\\
&= \int {dUdV\sum\limits_{p,q} {{\rm{Tr}}\left( {{U_{{p_1}}}V_{{p_1}}^\dag } \right){\rm{Tr}}\left( {{V_{{q_1}}}U_{{q_1}}^\dag } \right){\rm{Tr}}\left( {{U_{{p_2}}}V_{{p_2}}^\dag } \right){\rm{Tr}}\left( {{V_{{q_2}}}U_{{q_2}}^\dag } \right) \ldots {\rm{Tr}}\left( {{U_{{p_k}}}V_{{p_k}}^\dag } \right){\rm{Tr}}\left( {{V_{{q_k}}}U_{{q_k}}^\dag } \right)} }~.
\end{align}
There are many possible terms in those constructions. We could extract terms that could be written as frame potential for charge sectors, where each $UV$ terms are identified with $VU$. They look like
\begin{align}\label{sub}
{F^{(k)}} \supset \sum_p \#\int {dUdV {{\rm{Tr}}\left( {{U_{{p_1}}}V_{{p_1}}^\dag } \right){\rm{Tr}}\left( {{V_{{p_1}}}U_{{p_1}}^\dag } \right) \ldots {\rm{Tr}}\left( {{U_{{p_k}}}V_{{p_k}}^\dag } \right){\rm{Tr}}\left( {{V_{{p_k}}}U_{{p_k}}^\dag } \right)} }~.
\end{align}
Then the result is divided by partitions\footnote{Note that the expression for the spectral form factor also has a similar partitioning
\begin{align}
&R_{2k}^{\mathcal{E}} = \int {dU{{\left( {\sum\limits_{p,q} {{\rm{Tr}}\left( {{U_p}} \right){\rm{Tr}}\left( {U_q^\dag } \right)} } \right)}^k}} \nonumber\\
&= \int {dU\sum\limits_{p,q} {{\rm{Tr}}\left( {{U_{{p_1}}}} \right){\rm{Tr}}\left( {U_{{q_1}}^\dag } \right){\rm{Tr}}\left( {{U_{{p_2}}}} \right){\rm{Tr}}\left( {U_{{q_2}}^\dag } \right) \ldots {\rm{Tr}}\left( {{U_{{p_k}}}} \right){\rm{Tr}}\left( {U_{{q_k}}^\dag } \right)} }~. 
\end{align}
}, for instance, we have
\begin{align}
&F_{\cal E}^{(1)} \supset \sum\limits_p {F_{{{\cal E}_p}}^{(1)}}~,\nonumber\\
&F_{\cal E}^{(2)} \supset {\sum\limits_p {F_{{{\cal E}_p}}^{(2)}}  + 2\sum\limits_{p \ne q} {F_{{{\cal E}_p}}^{(1)}F_{{{\cal E}_q}}^{(1)}} }~,\nonumber\\
&F_{\cal E}^{(3)} \supset  {\sum\limits_p {F_{{{\cal E}_p}}^{(3)}}  + 9\sum\limits_{p \ne q} {F_{{{\cal E}_p}}^{(2)}F_{{{\cal E}_q}}^{(1)}}  + 6\sum\limits_{p \ne q \ne r} {F_{{{\cal E}_p}}^{(1)}F_{{{\cal E}_q}}^{(1)}F_{{{\cal E}_r}}^{(1)}} }~.\nonumber\\
&\cdots
\end{align}
For other terms, there is no naive conjugation that maintaining positivity in a single term, and cannot simply be written as frame potentials in charge sectors. 
\end{proof}
\subsubsection{Haar randomness}
For the U(1)-symmetric Haar system, frame potentials could simply be reduced to form factors.
\begin{thm}
\begin{align}
F_{{ \oplus _p}{{\cal H}_p}}^{(k)}=R_k^{{{ \oplus _p}{{\cal H}_p}}}~.
\end{align}
\end{thm}
\begin{proof}
Firstly we use
\begin{align}
R_{2k}^\mathcal{H} = F_\mathcal{H}^{(k)}~.
\end{align}
Secondly, for the Haar system, all terms that are outside of formula \ref{sub} vanishes. Namely, we have
\begin{align}
F_{{ \oplus _p}{{\cal H}_p}}^{(k)} = \sum\limits_p \# \int {dUdV {{\rm{Tr}}\left( {{U_{{p_1}}}V_{{p_1}}^\dag } \right){\rm{Tr}}\left( {{V_{{p_1}}}U_{{p_1}}^\dag } \right) \ldots {\rm{Tr}}\left( {{U_{{p_k}}}V_{{p_k}}^\dag } \right){\rm{Tr}}\left( {{V_{{p_k}}}U_{{p_k}}^\dag } \right)} } ~.
\end{align}
Combining with the previous analysis, we observe that
\begin{align}
F_{{ \oplus _p}{{\cal H}_p}}^{(k)}=R_k^{{{ \oplus _p}{{\cal H}_p}}}~.
\end{align}
Thus we could directly use the previous form factor results to predict frame potentials in the Haar system. 
\end{proof}
\subsubsection{$k$-invariant subspace}
Now we consider the case for $k$-invariant subspace. For terms that are inside \ref{sub}, the problem will be reduced to simplifying frame potentials in the single charge sector, which has already been computed in \cite{Cotler:2017jue}. For other terms, we give $F^{(1)}$ here as an example.
\begin{exmp}
For $F^{(1)}$ we know that
\begin{align}
F_{\cal E}^{(1)} = \sum\limits_{p \ne q} {\int {d{U_p}d{V_p}{\rm{Tr}}\left( {{U_p}V_p^\dag } \right)} \int {d{U_q}d{V_q}{\rm{Tr}}\left( {{V_q}U_q^\dag } \right)} }  + \sum\limits_p {F_{{{\cal E}_p}}^{(1)}} ~.
\end{align}
We have \cite{Cotler:2017jue}:
\begin{align}
F_{{{\mathcal E}_p}}^{(1)} = \left\{ {\begin{array}{*{20}{l}}
{1:{d_p} = 1}\\
{\frac{{R_2^{2,{{\cal E}_p}} + d_p^2 - 2R_2^{{{\cal E}_p}}}}{{d_p^2 - 1}}:{d_p} > 1}
\end{array}} \right.~.
\end{align}
The remaining terms are given by the Haar invariance,
\begin{align}
\int {d{U_p}d{V_p}{\rm{Tr}}\left( {{U_p}V_p^\dag } \right)}  = \frac{{\int {d{U_p}{\rm{Tr}}\left( {{U_p}} \right)} \int {d{V_p}} {\rm{Tr}}\left( {V_p^\dag } \right)}}{{{d_p}}} = \frac{{{{\left| {R_1^{{{\cal E}_p}}} \right|}^2}}}{{{d_p}}}~.
\end{align}
Similarly,
\begin{align}
\int {d{U_q}d{V_q}{\rm{Tr}}\left( {{V_q}U_q^\dag } \right)}  = \frac{{{{\left| {R_1^{\mathcal{E}_q}} \right|}^2}}}{{{d_q}}}~.
\end{align}
So we obtain
\begin{align}
F_{\cal E}^{(1)} = \sum\nolimits_q {\left\{ {\begin{array}{*{20}{l}}
{1:{d_q} = 1}\\
{\frac{{R_2^{2,{{\cal E}_p}} + d_q^2 - 2R_2^{{{\cal E}_p}}}}{{d_q^2 - 1}}:{d_q} > 1}
\end{array}} \right.}  + \sum\limits_{p \ne q} {\frac{1}{{{d_p}{d_q}}}{{\left| {R_1^{{{\cal E}_p}}} \right|}^2}{{\left| {R_1^{{{\cal E}_q}}} \right|}^2}}~.
\end{align}
\end{exmp}
We end this discussion by introducing the following simple observation.
\begin{thm}
If we assume that in the expansion of frame potential into charge sectors, the contribution from the highest form factor dominates, we have
\begin{align}
F_{\cal E}^{(k)} \approx \sum\limits_p {F_{{{\cal E}_p}}^{(k)}}  \approx \sum\limits_p {\frac{{R_{2k}^{2,{{\cal E}_p}}}}{{d_p^{2k}}}}~.
\end{align}
\end{thm}

\subsection{Decoupling}
Here we restrict our discussion to the decoupling property where the system is U(1)-symmetric Haar. 

We consider the system is factorized by subsystems $A$ and $B$. For a pure state $\rho_0$, we average over some unitary $U\rho_0U^\dagger$, and we consider the partial trace
\begin{align}
{\rho _A} = {\rm{T}}{{\rm{r}}_B}U{\rho _0}{U^\dag }~.
\end{align}
We compare the state $\rho_A$ and maximally mixed state on $A$,
\begin{align}
\Delta {\rho _A} = {\rho _A} - \frac{I}{{{d_A}}}~.
\end{align}
Generically, we have
\begin{align}
{\rm{Tr}}\left( {\Delta \rho _A^2} \right) = {\rm{Tr}}\left( {\rho _A^2} \right) - \frac{1}{{{d_A}}}~.
\end{align}
So the one-norm is bounded by
\begin{align}
\int_{} {dU\left\| {\Delta {\rho _A}} \right\|_1^2}  \le {d_A}\int_{} {dU\left\| {\Delta {\rho _A}} \right\|_2^2}  = \int {dU} {\rm{Tr}}\left( {\rho _A^2} \right) - 1~.
\end{align}
Thus, if we could bound $\int {dU} {\rm{Tr}}\left( {\rho _A^2} \right)$, we could then bound one-norm. 

Now, we prove the following theorem, which will be used for the discussion of the Hayden-Preskill experiment.
\begin{thm} [The U(1)-generalized Page theorem]
For large $d_q$, we have
\begin{align}
\int {dU} {\rm{Tr}}\left( {\rho _A^2} \right) \approx \frac{1}{{{d_q}({d_q} + 1)}}\left( {G({n_A},{n_B},q) + G({n_B},{n_A},q)} \right)~,
\end{align}
where
\begin{align}
G({n_A},{n_B},q) \equiv \sum\limits_{f = \max (0,q - {n_B})}^{\min ({n_A},q)} {d_f^{({n_A})}{{\left( {d_{q - f}^{({n_B})}} \right)}^2}}~.
\end{align}
\end{thm}
\begin{proof}
We firstly try to denote the expression in the following form,
\begin{align}
\int {dU} {\rm{Tr}}\left( {\rho _A^2} \right) = \left( {\int {dU} U_{{a_2}{b_2}}^{{a_1}{b_1}}U_{{a_4}{b_1}}^{\dag ,{a_3}{b_3}}U_{{{\tilde a}_2}{{\tilde b}_2}}^{{a_4}{{\tilde b}_1}}U_{{a_1}{{\tilde b}_1}}^{\dag ,{{\tilde a}_3}{{\tilde b}_3}}} \right)({\rho _0})_{{a_3}{b_3}}^{{a_2}{b_2}}({\rho _0})_{{{\tilde a}_3}{{\tilde b}_3}}^{{{\tilde a}_2}{{\tilde b}_2}}~.
\end{align}
Here the pair $ab$ means the combined basis in the subsystem $A$ and $B$. In order to proceed with the computation, we introduce further notations. We write the indices $a=(q_a,j_a)$, where $q_a$ specifies the charge sector while $j_a$ specifies the indices with fixed charge sector $q_a$. So we have
\begin{align}
(a,b) = ({q_a} + {q_b},({j_a},{j_b}))~.
\end{align}
The above expression looks very complicated. Thus here we only discuss a simpler situation, where we assume that the original state has fixed charge for subsystems $A$ and $B$, $q_A$ and $q_B$. Right now, $(a_2,b_2)$, $(a_3,b_3)$, $(\tilde{a}_2,\tilde{b}_2)$, $(\tilde{a}_3,\tilde{b}_3)$ are in the same charge sector $q=q_A+q_B$, since $U$ is given by a direct sum of different charge sectors, $(a_1,b_1)$, $(a_4,b_1)$, $(a_4,\tilde{b}_1)$ and $({a}_1,\tilde{b}_1)$ are still in the charge sector $q$, so we are free to use the Haar random formula in the charge sector $q$. For simplicity, we also assume that $n_A\le n_B$.

Using the Haar randomness formula, we have four terms. Two of them are contractions between $U$ and $U^\dagger$, while the other two of them are swaps. We only write the derivation in detail for the first term as an example, where $U$s are contracting with the nearest $U^\dagger$s. The rest of them are easy to generalize.

For the first term we have
\begin{align}
&\left( {\int {dU} U_{{a_2}{b_2}}^{{a_1}{b_1}}U_{{a_4}{b_1}}^{\dag ,{a_3}{b_3}}U_{{{\tilde a}_2}{{\tilde b}_2}}^{{a_4}{{\tilde b}_1}}U_{{a_1}{{\tilde b}_1}}^{\dag ,{{\tilde a}_3}{{\tilde b}_3}}} \right)({\rho _0})_{{a_3}{b_3}}^{{a_2}{b_2}}({\rho _0})_{{{\tilde a}_3}{{\tilde b}_3}}^{{{\tilde a}_2}{{\tilde b}_2}}\nonumber\\
&\supset \frac{1}{{d_q^2 - 1}}\delta _{{a_4}{b_1}}^{{a_1}{b_1}}\delta _{{a_2}{b_2}}^{{a_3}{b_3}}\delta _{{a_1}{{\tilde b}_1}}^{{a_4}{{\tilde b}_1}}\delta _{{{\tilde a}_2}{{\tilde b}_2}}^{{{\tilde a}_3}{{\tilde b}_3}}({\rho _0})_{{a_3}{b_3}}^{{a_2}{b_2}}({\rho _0})_{{{\tilde a}_3}{{\tilde b}_3}}^{{{\tilde a}_2}{{\tilde b}_2}}\nonumber\\
&= \frac{1}{{d_q^2 - 1}}\delta _{{a_4}{b_1}}^{{a_1}{b_1}}\delta _{{a_1}{{\tilde b}_1}}^{{a_4}{{\tilde b}_1}}({\rho _0})_{{a_2}{b_2}}^{{a_2}{b_2}}({\rho _0})_{{{\tilde a}_2}{{\tilde b}_2}}^{{{\tilde a}_2}{{\tilde b}_2}}~.
\end{align}
For the $\rho$ part, we know that scanning over indices separately in $A$ and $B$ is equivalently scanning the indices for the whole system, and sum over them, we just get 1 because it is the trace. The remaining part is equal to 
\begin{align}
&\frac{1}{{d_q^2 - 1}}\delta _{{a_4}{b_1}}^{{a_1}{b_1}}\delta _{{a_1}{{\tilde b}_1}}^{{a_4}{{\tilde b}_1}}\nonumber\\
&= \frac{1}{{d_q^2 - 1}}\sum\limits_{{b_1}} \begin{array}{l}
{\text{possibility of  }}{b_1}\\
{\text{from charge 0 to }}q
\end{array} \sum\limits_{{{\tilde b}_1}} \begin{array}{l}
{\text{possibility of  }}{{\tilde b}_1}\\
{\text{for the same charge as }}{b_1}
\end{array} \sum\limits_{{a_1}} \begin{array}{l}
{\text{possibility of  }}{a_1}\\
q(a_1)=q - q({b_1})
\end{array}~,
\end{align}
where $q(a)$ will give the number of charge sector for indices $a$, and the result is
\begin{align}
&{\rm{If: }}q \le {n_A} \le {n_B}:\frac{1}{{d_q^2 - 1}}\sum\limits_{q\left( {{b_1}} \right) = 0}^q {d_{q - q\left( {{b_1}} \right)}^{({n_A})}{{\left( {d_{q\left( {{b_1}} \right)}^{({n_B})}} \right)}^2}} {\rm{ }}~,\nonumber\\
&{\rm{If: }}{n_A} \le q \le {n_B}:\frac{1}{{d_q^2 - 1}}\sum\limits_{q\left( {{b_1}} \right) = q - {n_A}}^q {d_{q - q\left( {{b_1}} \right)}^{({n_A})}{{\left( {d_{q\left( {{b_1}} \right)}^{({n_B})}} \right)}^2}} ~,\nonumber\\
&{\rm{If: }}{n_A} \le {n_B} \le q:\frac{1}{{d_q^2 - 1}}\sum\limits_{q\left( {{b_1}} \right) = q - {n_A}}^{{n_B}} {d_{q - q\left( {{b_1}} \right)}^{({n_A})}{{\left( {d_{q\left( {{b_1}} \right)}^{({n_B})}} \right)}^2}} ~.
\end{align}
One could show that for the swap terms, there is nothing but an extra smaller factor $ -\frac{1}{d_q({d_q^2 - 1})}$ instead of $\frac{1}{{d_q^2 - 1}}$. Thus swap terms are less dominated in the case of large $d_q$. Thus we conclude 
\begin{align}
\int {dU} {\rm{Tr}}\left( {\rho _A^2} \right) = \frac{1}{{{d_q}({d_q} + 1)}}\left( {G({n_A},{n_B},q) + G({n_B},{n_A},q)} \right)~,
\end{align}
where
\begin{align}
G({n_A},{n_B},q) \equiv \sum\limits_{f = \max (0,q - {n_B})}^{\min ({n_A},q)} {d_f^{({n_A})}{{\left( {d_{q - f}^{({n_B})}} \right)}^2}} ~.
\end{align}
\end{proof}
A similar expression is derived by \cite{Yoshida:2018ybz} in some cases of the U(1)-symmetric Hayden-Preskill experiment, which we will discuss later.

\section{Example: the complex SYK model}\label{EG}
Now we discuss a standard example, the complex SYK model, a very good candidate for approximate charge decoupling and $k$-invariant subspace. 

The complex SYK model is given by the following Hamiltonian
\begin{align}
H = \sum\limits_{i,j,k,l} {{J_{i,j,k,l}}f_i^\dag f_j^\dag {f_k}{f_l}}~,
\end{align}
where $f$s are Dirac fermions ($f$ and $f^\dagger$ are the annihilation and creation operators respectively). $J$ is given by independent complex Gaussian distribution with constraints:
\begin{align}
&{J_{ijkl}} =  - {J_{jikl}}~,~~~~~~{J_{ijkl}} =  - {J_{ijlk}}~,\nonumber\\
&{J_{ijkl}} = J_{klij}^*~,~~~~~~\left\langle {{{\left| {{J_{ijkl}}} \right|}^2}} \right\rangle  = \frac{{4{J^2}}}{{{N^3}}}~,
\end{align}
and with zero mean. Sometimes we also include a fermion mass term $\sum _i m_f f_i^\dagger f_i$ but here we set the mass $m_f=0$. By construction, in the fermion number zero and one sector of this model, we have zero energy eigenvalues. So the spectrum in those fermionic charge sectors is not random. 

We plot the density of states for some single charge sectors and the whole sector for the $N=14$ complex SYK model in Figure \ref{dssecs} and Figure \ref{dswhole}, respectively, where we shift the energy such that $E=0$ is the ground state in each sector. A clear feature of those plots is the edge near the ground state, where for the single charge sector, we get $\rho(E)\sim E^{1/2}$, and for the whole sector, we get $\rho(E)\sim E$. This is a feature that is pointed out by a series of works \cite{Sachdev:2019bjn,Davison:2016ngz,TP,Liu:2019niv}. The square root edge of $\rho(E)\sim E^{1/2}$ is consistent with the Gaussian random matrix theory and Schwarzian quantum mechanics, while the linear edge $\rho(E)\sim E$ is from an extra contribution of U(1) phase. In each charge sector (except fermionic number zero and one), the energy spectrum distribution, around the low energy, could be described by a Gaussian random matrix theory. From the classification in \cite{You:2016ldz}, the case $N=14$ corresponds to Gaussian random unitary GUE.
\begin{figure}[htbp]
  \centering
  \includegraphics[width=0.6\textwidth]{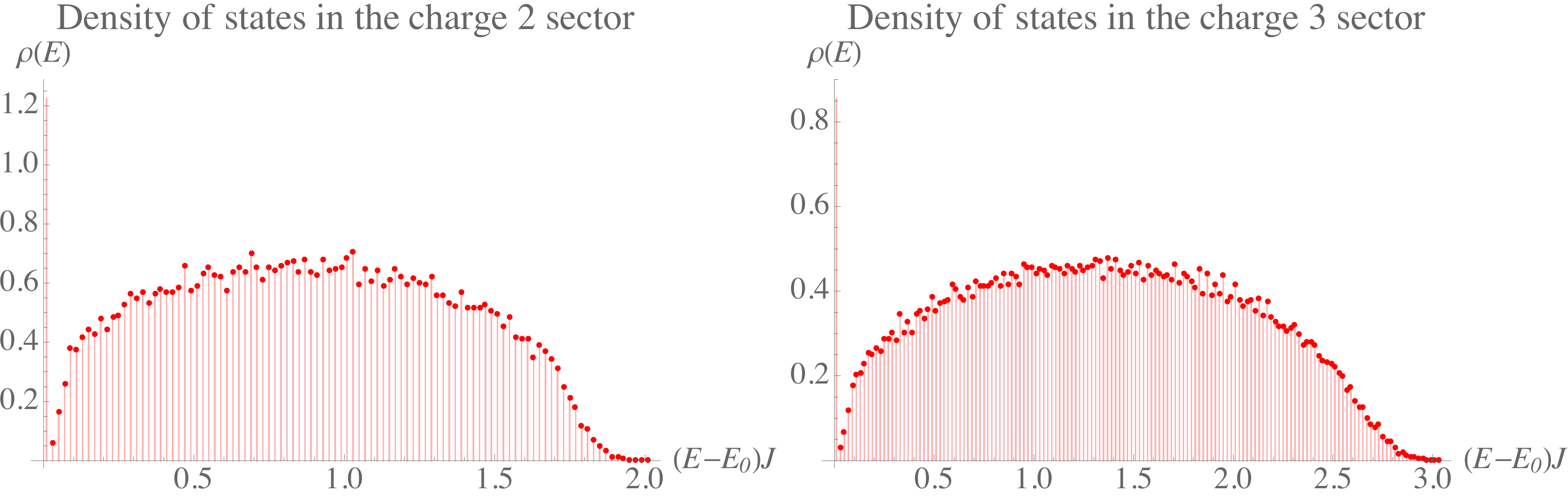}
  \includegraphics[width=1.0\textwidth]{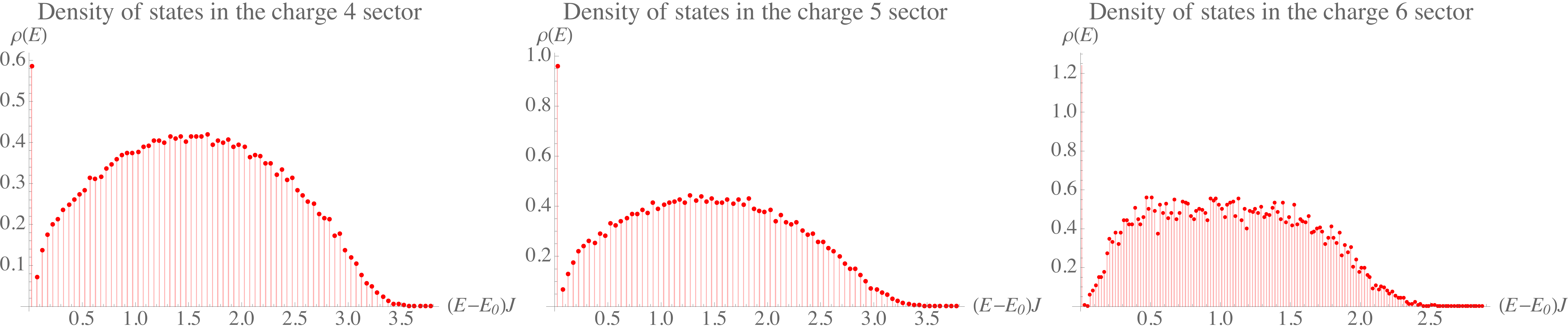}
  \caption{\label{dssecs} Density of states in different sectors for the $N=14$ complex SYK model. We use 2000 random realizations. Those plots clearly show a square root edge near the ground state, which is a standard consequence of random matrix theory and Schwarzian quantum mechanics.}
\end{figure}
\begin{figure}[htbp]
  \centering
  \includegraphics[width=0.5\textwidth]{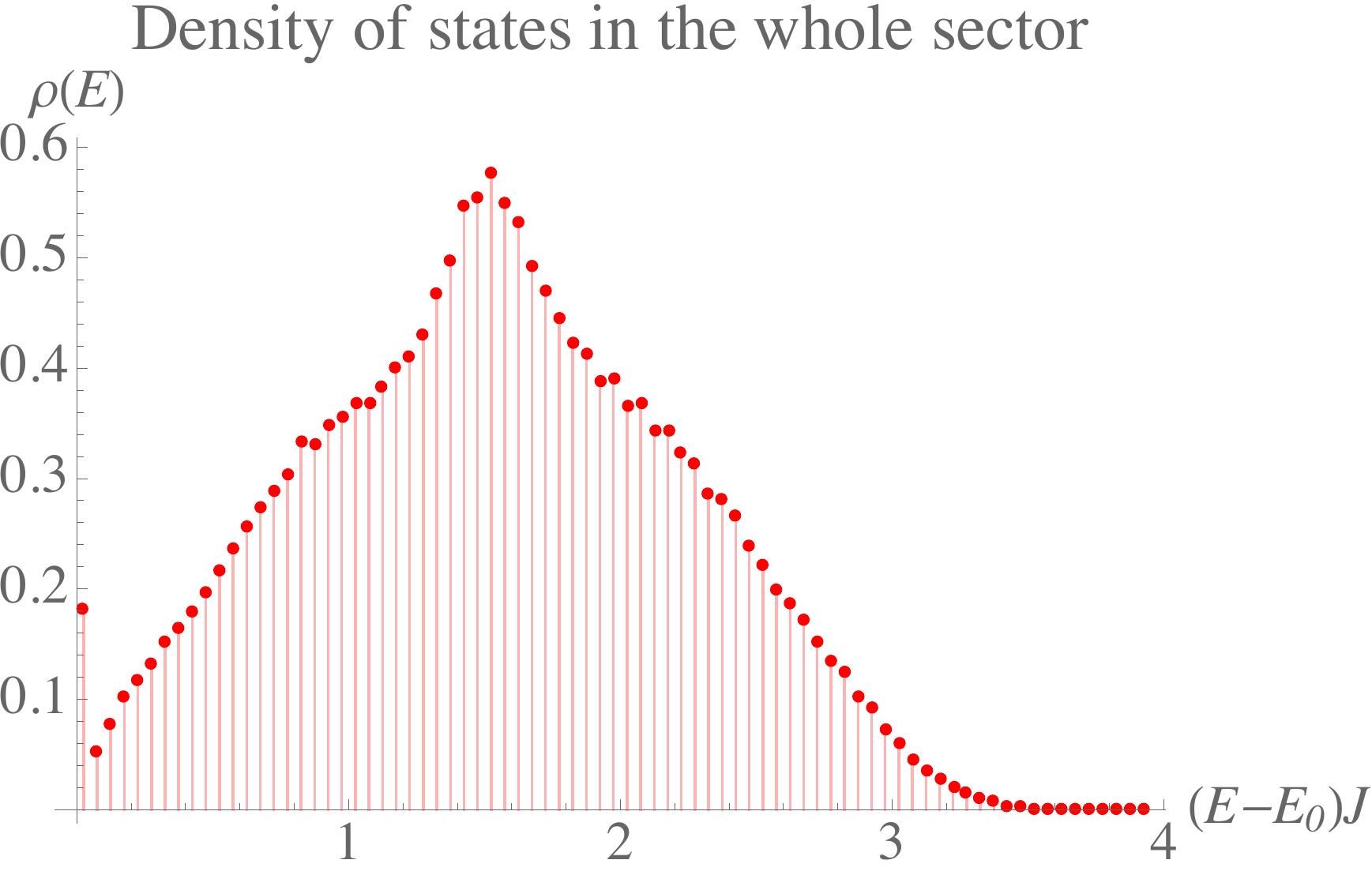}
  \caption{\label{dswhole} Density of states in the whole sector for the $N=14$ complex SYK model. We use 2000 random realizations. This plot clearly shows a linear edge around the ground state. This could be obtained by a usual Schwarzian contribution times an extra U(1) phase \cite{Sachdev:2019bjn,Davison:2016ngz,TP,Liu:2019niv}.}
\end{figure}

One could also study spectral form factor in such a theory in infinite temperature. Generically in the chaotic system, the spectral form factor starts from a slope down to a dip, and then a ramp towards a flat plateau. We plot the spectral form factor in Figure \ref{sffsecs} for each charge sector, in Figure \ref{sffwhole} for the whole charge sector. From these results, one could verify the formula 
\begin{align}\label{R2secs}
R_2^{{ \oplus _p}{{\cal E}_p}} = \sum\limits_p {R_2^{{{\cal E}_p}}({d_p})}  + \sum\limits_{p \ne q} {R_1^{{{\cal E}_p}}\left( {{d_p}} \right)R_1^{{{\cal E}_q}*}\left( {{d_q}} \right)}~,
\end{align}
approximately holds. The relative error is given in Figure \ref{reR2}. The small deviation from this formula is because of the hidden correlation between energy eigenvalues of different charge sectors. Thus, this indicates a small violation of charge decoupling. 

There is another feature that around the time where each sector is approaching the Haar randomness, there is a peak on the relative error. This feature shows that the correlation between charge sectors gets relatively large around the scrambling time.
\begin{figure}[htbp]
  \centering
  \includegraphics[width=0.6\textwidth]{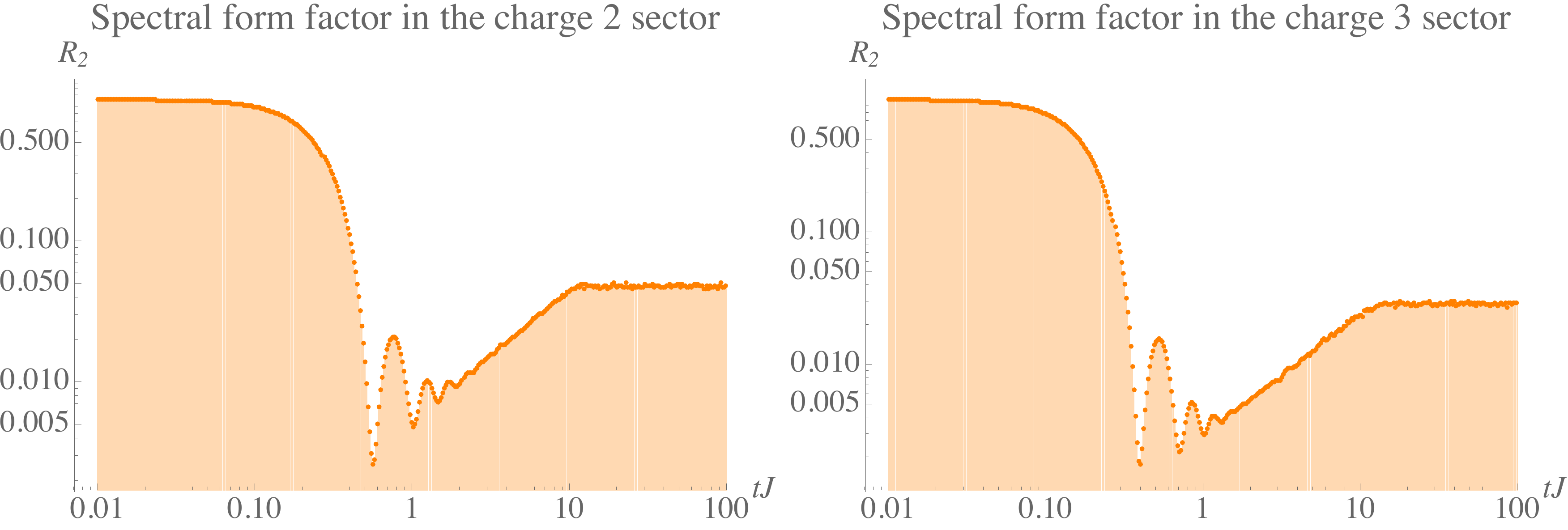}
  \includegraphics[width=1.0\textwidth]{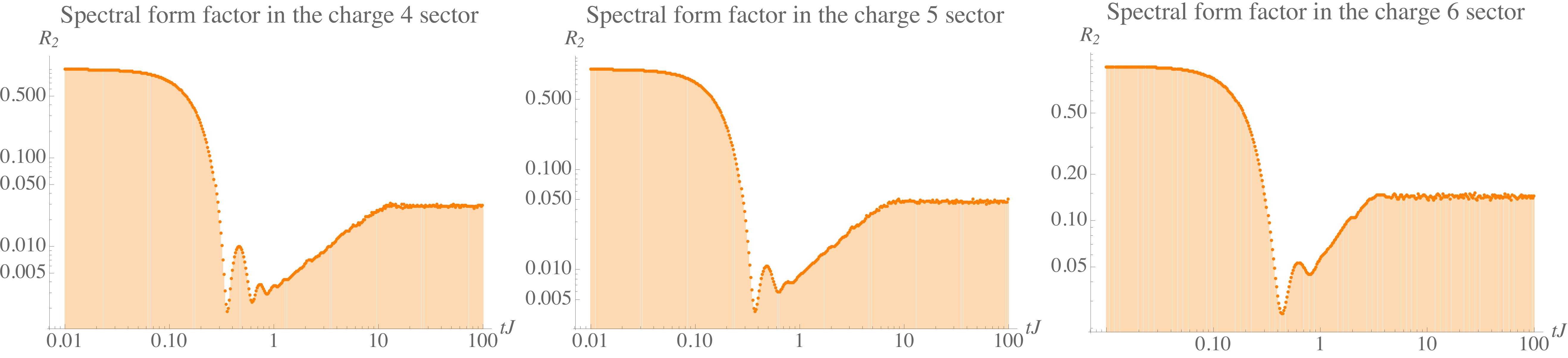}
  \caption{\label{sffsecs} Spectral form factor $R_2(t)$ in different sectors for the $N=14$ complex SYK model. We use 2000 random realizations.}
\end{figure}
\begin{figure}[htbp]
  \centering
  \includegraphics[width=0.6\textwidth]{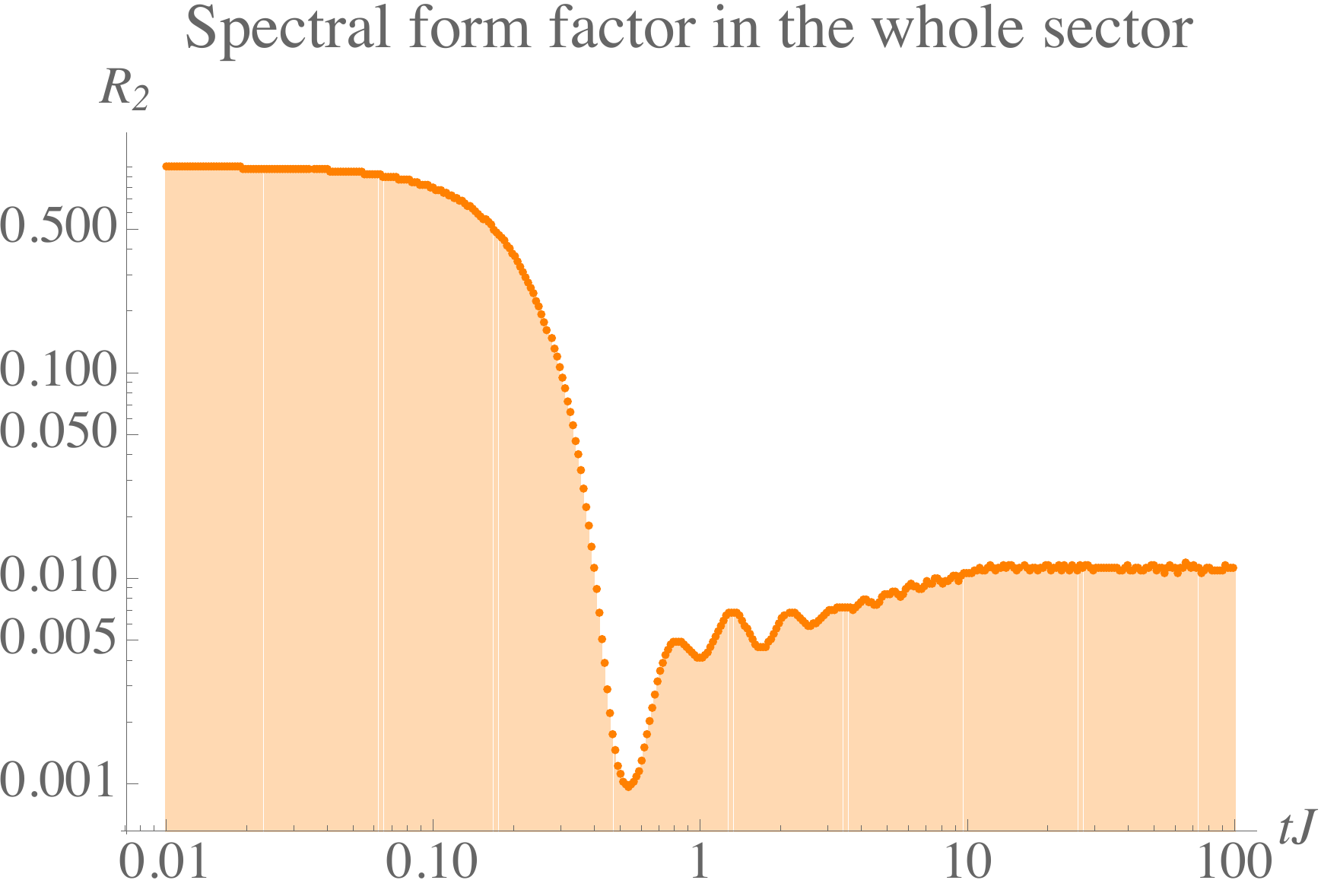}
  \caption{\label{sffwhole} Spectral form factor $R_2(t)$ in the whole sector for the $N=14$ complex SYK model. We use 2000 random realizations.}
\end{figure}
\begin{figure}[htbp]
  \centering
  \includegraphics[width=0.6\textwidth]{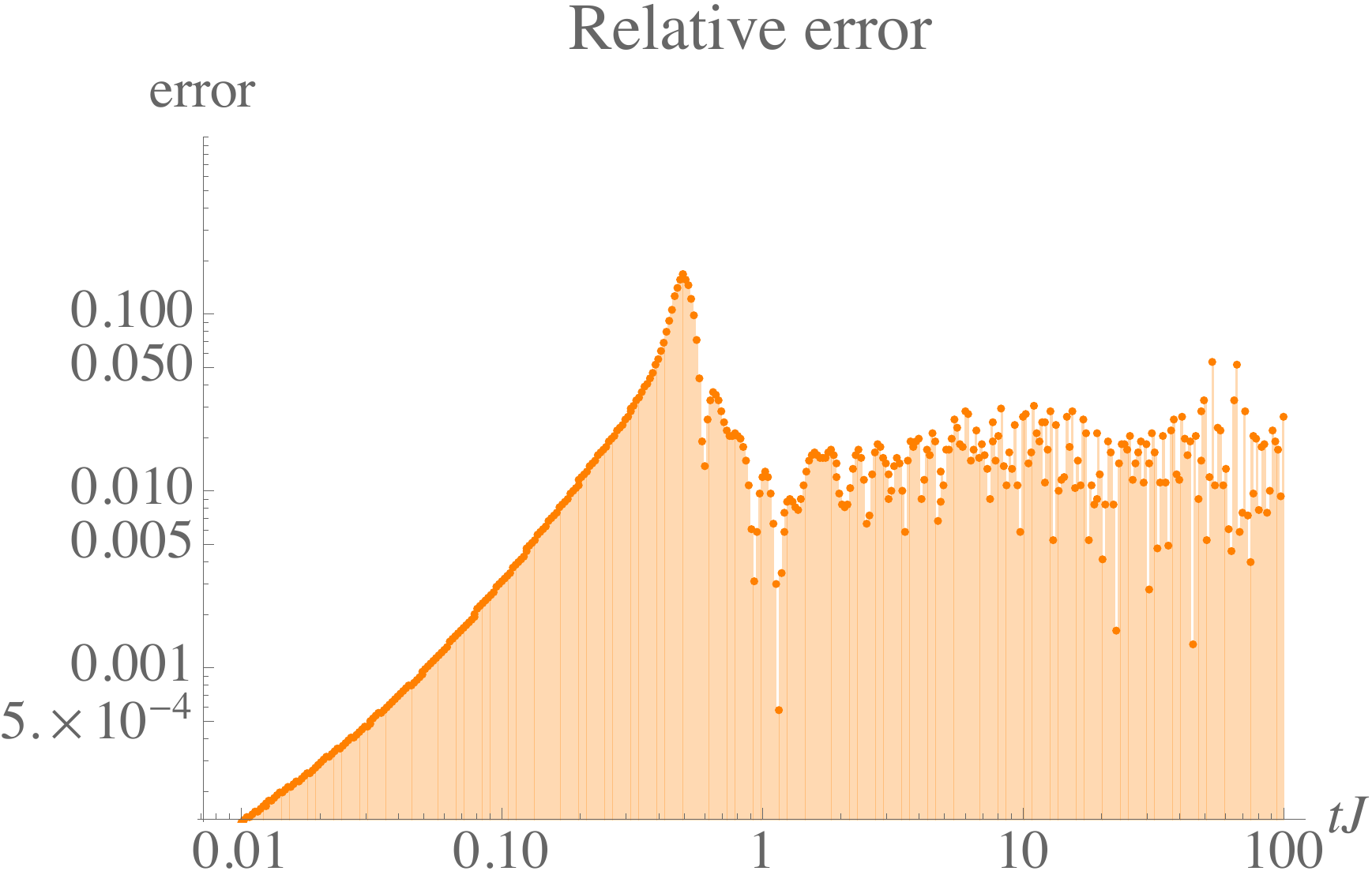}
  \caption{\label{reR2} The relative error comparing two sides of the formula \ref{R2secs} for the $N=14$ complex SYK model. We use 2000 random realizations. The relative error is always smaller than 0.2.}
\end{figure}

Similarly, one can study the frame potential. In the GUE and the Majonara SYK model, we expect frame potential should decay towards a dip, and stay in the dip for a while, and then grows through a ramp, finally towards a plateau. In Figure \ref{fpsecs} and Figure \ref{fpwhole} we give plots for $k=1$ frame potentials as example. Frame potentials in $q=2\sim 6$ sectors are very close to systems with the Haar invariance. (In fact, due to random matrix theory classification mentioned before, each sector is predicted by GUE level statistics, where GUE is an exact Haar-invariant system and its $k$-invariance is zero for all $k$). We compute $1$-invariance in $q=2\sim 6$ sectors in Figure \ref{1invsecs}. 

Moreover, we check the following formula
\begin{align}\label{F1secs}
F_{\cal E}^{(1)} =\sum\limits_p F_{\mathcal{E}_p}^{(1)}  + \sum\limits_{p \ne q} {\frac{1}{{{d_p}{d_q}}}{{\left| {R_1^{{{\cal E}_p}}} \right|}^2}{{\left| {R_1^{{{\cal E}_q}}} \right|}^2}}~,
\end{align}
by plotting the relative error between two sides of the equation in Figure \ref{reF1}. This analysis shows that the error from the prediction \ref{F1secs} is small (smaller than 0.1), while the mismatch is again, due to correlations between energy eigenvalues in different charge sectors, namely a small violation of charge decoupling. 
\begin{figure}[htbp]
  \centering
  \includegraphics[width=0.6\textwidth]{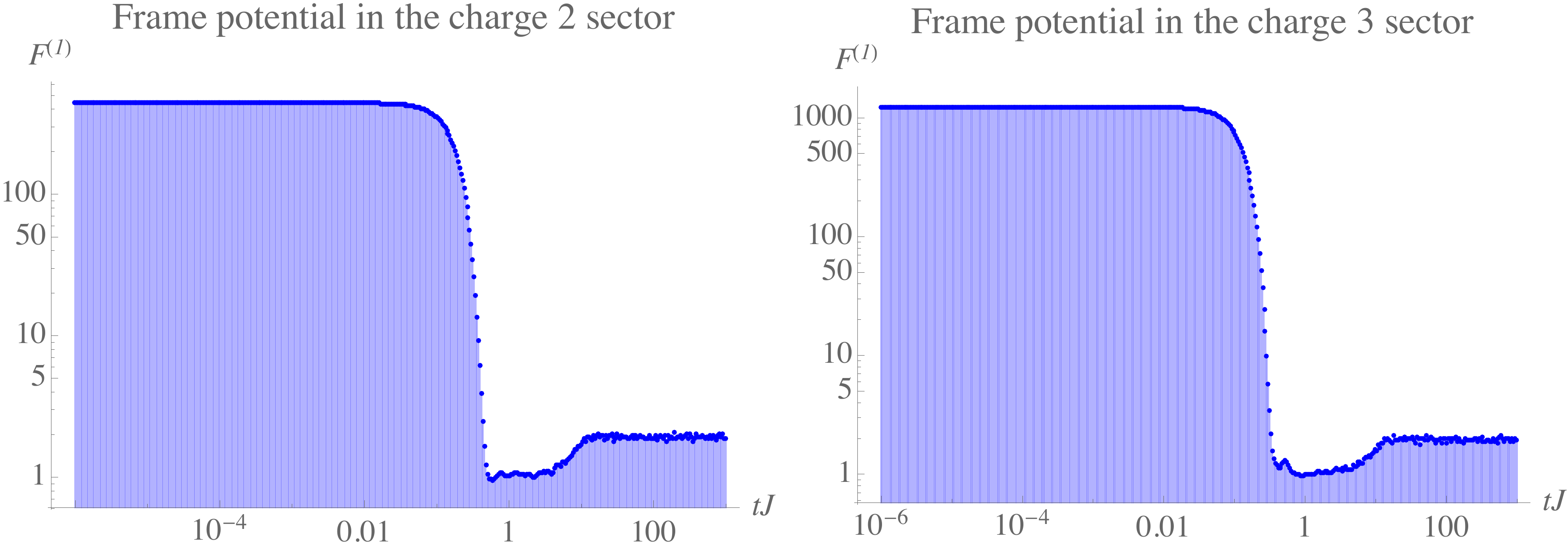}
  \includegraphics[width=1.0\textwidth]{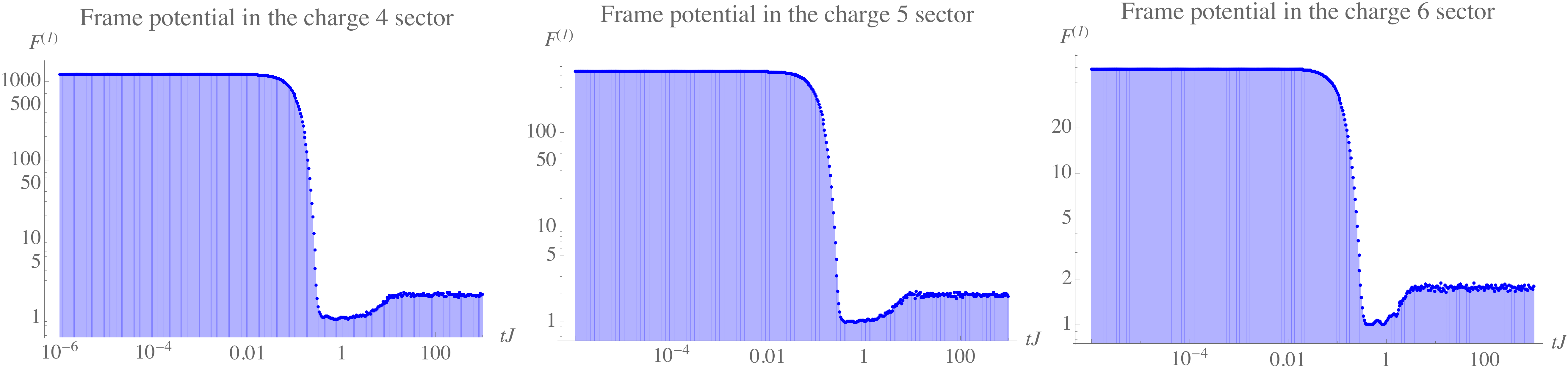}
  \caption{\label{fpsecs} Frame potential $F^{(1)}$ in different sectors for the $N=14$ complex SYK model. We use 1500 random realizations.}
\end{figure}
\begin{figure}[htbp]
  \centering
  \includegraphics[width=0.6\textwidth]{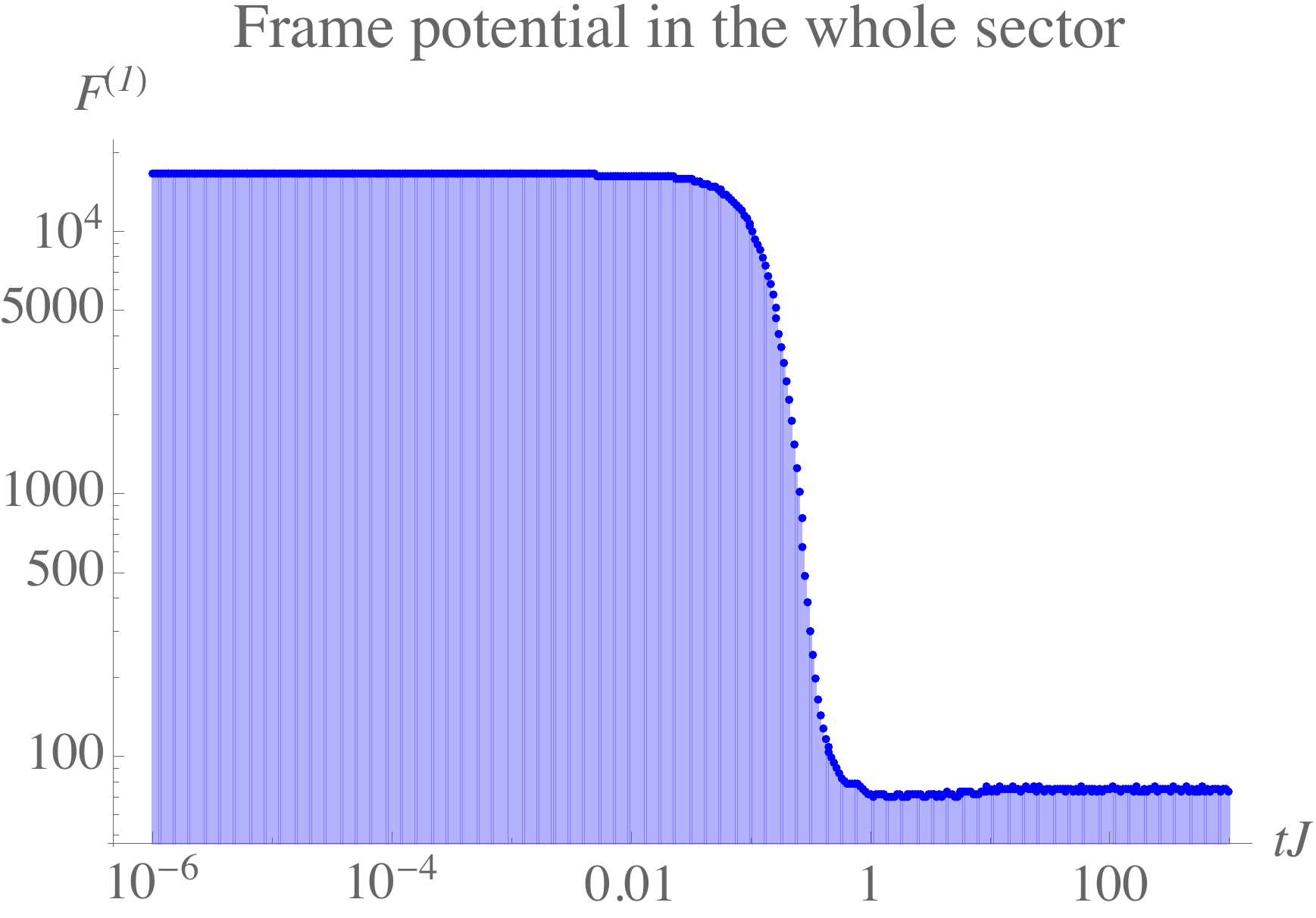}
  \caption{\label{fpwhole} Frame potential $F^{(1)}$ in the whole sector for the $N=14$ complex SYK model. We use 1500 random realizations.}
\end{figure}
\begin{figure}[htbp]
  \centering
  \includegraphics[width=0.6\textwidth]{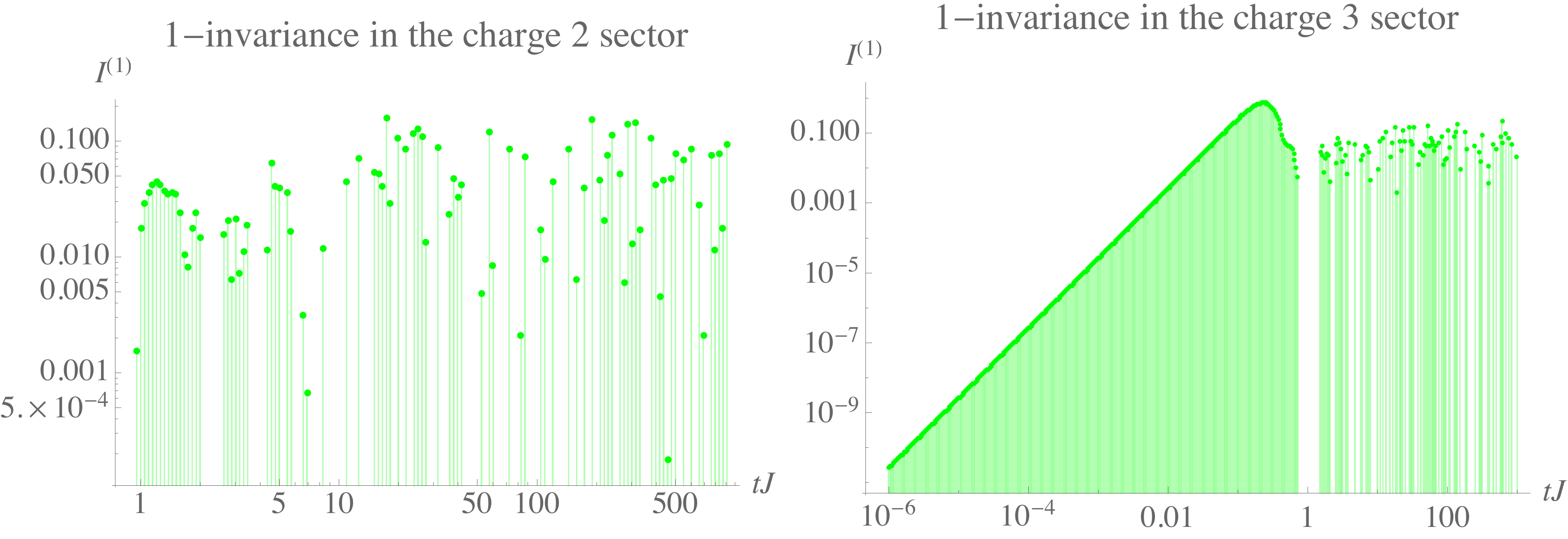}
  \includegraphics[width=1.0\textwidth]{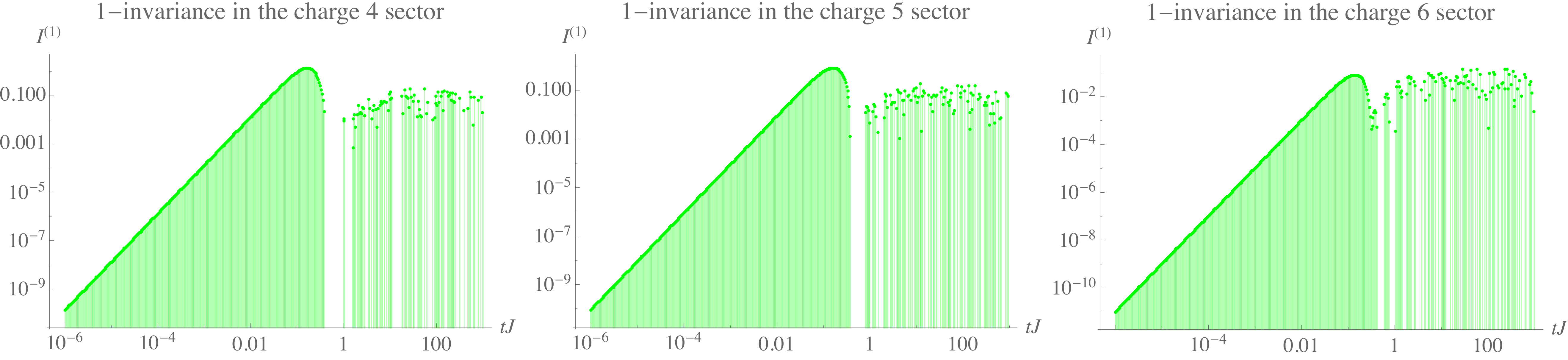}
  \caption{\label{1invsecs} 1-invariance $I^{(1)}$ in different sectors for the $N=14$ complex SYK model. We use 1500 random realizations.}
\end{figure}
\begin{figure}[htbp]
  \centering
  \includegraphics[width=0.6\textwidth]{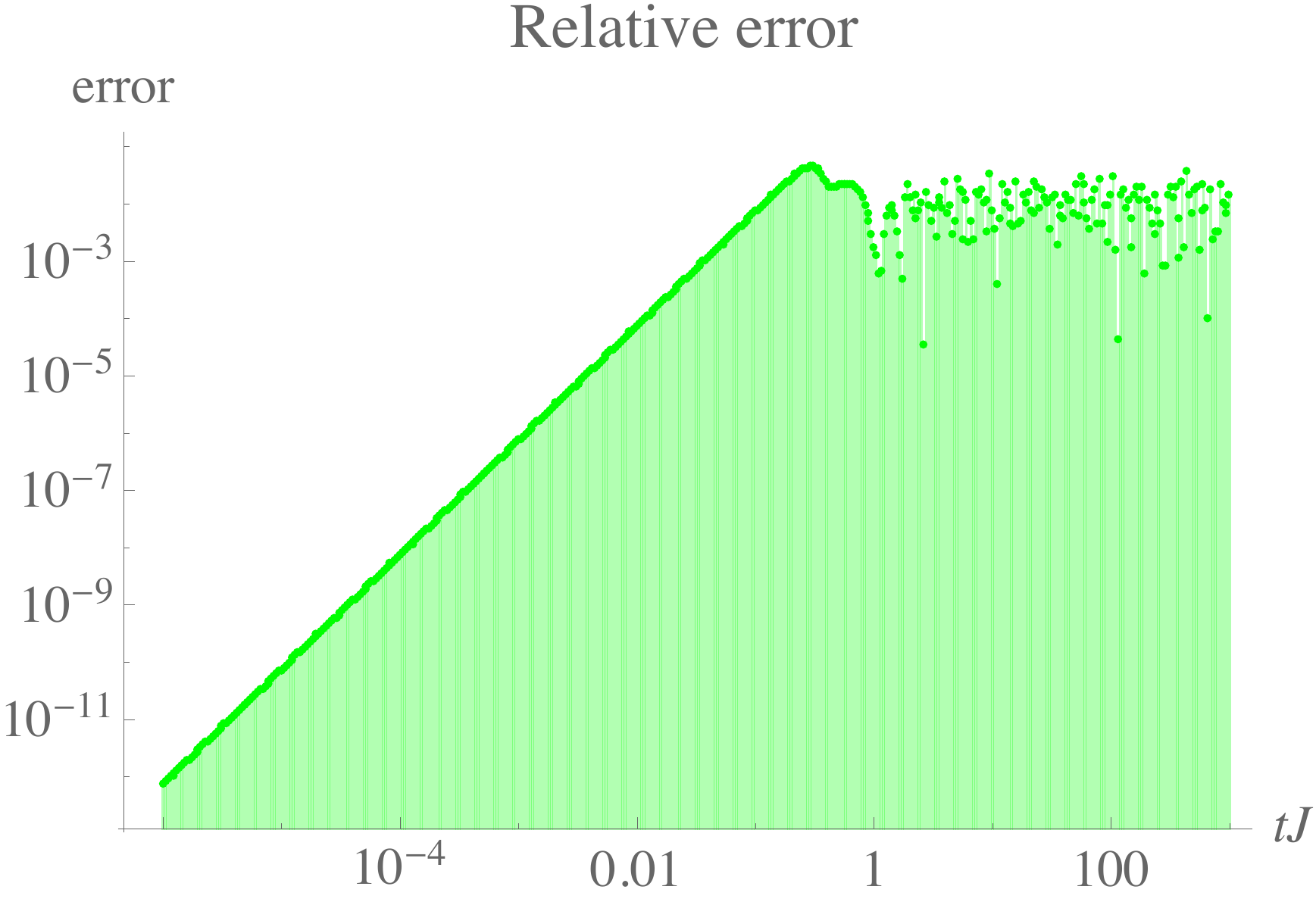}
  \caption{\label{reF1} The relative error comparing two sides of the formula \ref{F1secs} for the $N=14$ complex SYK model. The relative error is always smaller than 0.1.}
\end{figure}

One another feature we could obtain is the time dependence. From \ref{reR2} and \ref{reF1} we see the deviation keeps growing until the system gets sufficiently scrambled, which implies that during scrambling, the correlation between different charge sectors gets amplified. We also observe from \ref{1invsecs} that 1-invariance keeps growing until scrambling, which means that in each sector, during scrambling, we get more deviation from exact invariance under the Haar transformation, although measured in frame potential, the system gets more closed to the Haar ensemble $\mathcal{H}$ following \ref{fpsecs}.

\section{Codes}\label{CO}
This section is written for quantum error correction interpretation of random unitaries. In this part, we will describe the theory of quantum error correction for the Haar randomness and its U(1) extension. A standard application of the above theory is the Hayden-Preskill experiment. For the U(1)-symmetric Haar randomness, we show that it is consistent with the upper bound on the fidelities for covariant codes. Finally, we will address some possible implications for black hole thought experiments.
\subsection{Quantum error correction for the Haar randomness}
The phenomenon of scrambling is widely used in the construction of quantum error correction codes, which could be understood as a connection between quantum error correction and chaos. 

As a warmup example, we consider the following construction. Consider the system $A$ and $B$, we put $A$ with state basis $\ket{\alpha_a}$ with $1\le a \le d_A$, and $B$ with a fixed state $\ket{B}$ ($d_A$ and $d_B$ are dimensions of $A$ and $B$). We define the code subspace,
\begin{align}
{\cal C} = \left\{ \beta_a\equiv {U\left| {{\alpha _a},B} \right\rangle :a = 1,2, \ldots {d_A}} \right\}~.
\end{align}
Here, $U$ is a sample from the Haar random ensemble with dimension $d_A+d_B$. We understand $U$ as the encoding map. Thus, the dimension of the code subspace $\mathcal{C}$ is the same as the dimension of $A$. 

In order to show the ability of quantum error correction of the above code, we try to do some heuristic computations about the Knill-Laflamme condition, and we take the average of it. We consider
\begin{align}
\left\langle {{\beta _a}} \right|O\left| {{\beta _b}} \right\rangle  = \left\langle {{\alpha _a},B} \right|{U^\dag }OU\left| {{\alpha _b},B} \right\rangle~,
\end{align}
where $O$ is a generic operator. We could compute the Haar average
\begin{align}
&\int_{\cal H} {dU} \left\langle {{\beta _a}} \right|O\left| {{\beta _b}} \right\rangle= \int_{\cal H} {dU} \left\langle {{\alpha _a},B} \right|{U^\dag }OU\left| {{\alpha _b},B} \right\rangle  \nonumber\\
&= \int_{\cal H} {dU} U_i^{\dag a}O_j^iU_b^j= \frac{1}{L}\delta _b^a\delta _i^jO_j^i = \delta _b^a\left\langle O \right\rangle ~,
\end{align}
which means that this condition is satisfied for an arbitrary operator. This is impossible in the usual definition of exact quantum error correction code without averaging, since code parameters should satisfy some bounds (for instance, the quantum singleton bound or the hamming bound). 

Since it is a random average, we might also consider the variance of the Knill-Laflamme condition. Similarly, we compute 
\begin{align}
&\int_{\cal H} {dU} {\left| {\left\langle {{\beta _a}} \right|O\left| {{\beta _b}} \right\rangle } \right|^2} - {\left| {\int_{\cal H} {dU} \left\langle {{\beta _a}} \right|O\left| {{\beta _b}} \right\rangle } \right|^2}\nonumber\\
&= \int_{\cal H} {dU} U_i^{\dag a}O_j^iU_b^jU_{i'}^{\dag b}O_{j'}^{\dag i'}U_a^{j'} - \delta _b^a{\left| {\left\langle O \right\rangle } \right|^2}\nonumber\\
&= \frac{1}{{{L^2} - 1}}\left( {L{{\left| {\left\langle O \right\rangle } \right|}^2}\delta _b^a(L - 1) + \left\langle {O{O^\dag }} \right\rangle (L - 1)} \right) - \delta _b^a{\left| {\left\langle O \right\rangle } \right|^2}\nonumber\\
&= \frac{{L{{\left| {\left\langle O \right\rangle } \right|}^2}\delta _b^a + \left\langle {O{O^\dag }} \right\rangle }}{{L + 1}} - \delta _b^a{\left| {\left\langle O \right\rangle } \right|^2}\nonumber\\
&=  - \frac{{{{\left| {\left\langle O \right\rangle } \right|}^2}}}{{L + 1}}\delta _b^a + \frac{{\left\langle {O{O^\dag }} \right\rangle }}{{L + 1}}~.
\end{align}
This computation shows that the variance is approaching zero for a large system $L\to \infty$. As a clearer example, we consider the Pauli error. We set $O = P_\mu ^\dag {P_\nu}$. It is easy to show that
\begin{align}
&\int_{\mathcal{H}} d U\left\langle {{\beta _a}|O|{\beta _b}} \right\rangle  = {\delta _{ab}}{\delta _{\mu \nu }}\nonumber\\
&\int_{\mathcal{H}} d U{\left| {\left\langle {{\beta _a}|O|{\beta _b}} \right\rangle } \right|^2} - {\left| {\int_{\mathcal{H}} d U\left\langle {{\beta _a}|O|{\beta _b}} \right\rangle } \right|^2} = \frac{{1 - \delta _b^a{\delta _{\mu \nu }}}}{{L + 1}}~.
\end{align}

Another aspect of showing quantum error correction ability of the Haar randomness, is through the discussion of decoupling. It follows from the following simple facts. Consider a distance $t+1$ code. Say that we have a state $\ket{\psi}$ in the code subspace with dimension $n$ and we trace out $n-t$ qubits,
\begin{align}
{\rho ^{(t)}} = {\rm{T}}{{\rm{r}}_{(n - t)}}\left| \psi  \right\rangle \left\langle \psi  \right|~.
\end{align}
If it is a \emph{non-degenerate} code, which means that the matrix $C_{ab}$ in the Knill-Laflamme condition is not degenerate, for any error operator with weight up to $t$. Then one can show that
\begin{align}
{\rho ^{(t)}} = \frac{{{I_t}}}{{{2^t}}}~.
\end{align}
In fact, non-degeneracy implies that 
\begin{align}
{\rm{Tr}}\left( {{\rho ^{(t)}}O} \right) = 0~,
\end{align}
if $O$ is not identity, since one can expand $\rho^{(t)}$ as
\begin{align}
{\rho ^{(t)}} = \frac{1}{{{2^t}}}{I_t} + \sum\limits_O {{\rho _O}O}~,
\end{align}
where the sum is taken over Paulis. Then, from the orthogonal property of Pauli operators, we have all $\rho_O=0$, which implies that $\rho$ is maximally mixed. This is another evidence supporting the fact that the Haar random unitaries have good error correction properties, following from the decoupling property we have discussed above.

One might note that maybe it is not very suitable to describe approximate, random quantum error correction codes in terms of code parameters like the distance. Thus, people usually rely on the definition of the fidelity of recovery to define the decoding capability properly. This might be particularly important for applications in holography, since the concept of quantum error corrections usually hold approximately due to the leading order, the semi-classical gravitational path integral. The phenomenon of non-zero variance has also been shown in the recent discussions about the black hole information paradox and the baby universe, where the disordered average is shown explicitly in some examples of gravitational path integral and holographic dual (for instance, see some decent discussions \cite{Penington:2019npb,Almheiri:2019psf,Penington:2019kki,Almheiri:2019qdq,Marolf:2020xie}.).

The following theorem will somewhat make the above statement rigorous \cite{recovery,Sutter:2016uxq},
\begin{thm}\label{recover}
Consider a mixed state on the system $YZ$, which is purified by $X$. The whole state $\rho_{XYZ}$ is a pure state. We assume that there are noises only acting on $Z$. Say that there are recovery maps $\mathcal{R}=\mathcal{D} \circ \mathcal{N}$ acting on $\rho_{XY}$, then the supremum of the state fidelity over all possible recoveries (the fidelity of recovery) is bounded by the conditional mutual information
\begin{align}
I{(X:Z|Y)_\rho } \ge  - 2{\log _2}\mathop {\sup }\limits_{\mathcal{R}} F\left( {{\rho _{XYZ}},\mathcal{R}\left( {{\rho _{XY}}} \right)} \right) \equiv  - 2{\log _2}{F_\rho }~,
\end{align}
where the fidelity is defined by
\begin{align}
F(\rho ,\sigma ) = {\left\| {\sqrt \rho  \sqrt \sigma  } \right\|_1}~,
\end{align}
and we have the definition of the conditional mutual information
\begin{align}
I{(X:Z|Y)_\rho } \equiv S({\rho _{XY}}) + S({\rho _{YZ}}) - S({\rho _Y}) - S({\rho _{XYZ}})~,
\end{align}
where $S$ is the von Neumann entropy.
\end{thm}
One such recovery map with high fidelity is called the (twirled) Petz map \cite{petz}. The condition where the conditional mutual information is zero, is called the quantum Markov condition for states. For chaotic randomness, like the Haar distribution, the state is highly chaotic and thus Markov, ensuring a potential high fidelity of recoveries. The above condition is also used in the context of holography, see, for instance, \cite{Pastawski:2016ggn,Pastawski:2016qrs,Lewkowycz:2019xse}.

We will show later that the above condition is exactly the decoupling condition in the Hayden-Preskill experiment. 
\subsection{Error correction in the Hayden-Preskill experiment}
The above picture is closely related to the Hayden-Preskill experiment. We will do the Haar case with a more detailed analysis here. The Hayden-Preskill experiment is used to model a black hole and discuss information transfer during Hawking radiation, which is described in Figure \ref{fig1}. In this diagram, Alice $A$ has a small amount of information (a state), while $B$ is a black hole that is connected (maximally entangled) with Bob's quantum computer $\bar{B}$. $\bar{A}$ is a reference system, forming a Bell pair with $A$. After time evolution, $C$ becomes the remaining black hole, and $D$ is the Hawking radiation. 

One could identify the above construction in the Theorem \ref{recover}. After the encoding $U$, the reference system $X$ is identified with $\bar{A}$, and we also set $D\bar{B}$ as $Y$ and $C$ as $Z$. The noise is understood as the erasure of $C$. Thus, this procedure could be understood exactly as the error correction process above. Alice's diary is understood as codewords, and random unitary $U$ is understood as an encoding map, while the error is understood as the erasure. The fact that random unitary could serve as a good error correction code, ensures that one could correct the error (erasure noise) and recover the original codewords (Alice's diary). Note that in this example, the codewords are understood as the states of Bell pairs within $A\bar{A}$, which has the Hilbert space dimension $d_A$. In such construction, the codewords are fixed, and we could choose the states in $A\bar{A}$ to be the computational basis, and we choose the random encoding $U$ as a random matrix with respect to this basis. Since we are taking the average over a given random state, Theorem \ref{recover} could apply\footnote{We might notice that the decoupling condition in the Hayden-Preskill experiment here is using the 2-Renyi entropy, while we are bounding the fidelity of recovery using the von Neumann entropy. When decoupling happens, the state is nearly maximally mixed. Thus we could ignore such a difference. However, we might also consider some general bounds from the 2-Renyi conditional mutual information towards the von Neumann conditional mutual information. (see some related works \cite{Renyi,Fer,Dupuis:2015vxa}, or the appendices of \cite{Yoshida:2018vly}).}.

After applying the unitary operator, the state is given by
\begin{align}
\left| \Psi  \right\rangle  = \frac{1}{{\sqrt {{d_A}{d_B}} }}U_{cd}^{ab}\left| {a,b,c,d} \right\rangle~,
\end{align}
where we use $a,b$ to denote the basis from $\bar{A}$ and $\bar{B}$, while $c,d$ to denote the basis from $C$ and $D$. So we compute the reduced density matrix
\begin{align}
&{\rho _{\bar{A}C}} = \frac{1}{{{d_A}{d_B}}}U_{cd}^{ab}U_{\tilde ab}^{\dag \tilde cd}\left| {a,c} \right\rangle \langle \tilde a,\tilde c|~,\nonumber\\
&\frac{I}{{{d_A}}} \otimes {\rho _C} = \frac{1}{{d_A^2{d_B}}}U_{cd}^{ab}U_{ab}^{\dag \tilde cd}|\tilde a,c\rangle \langle \tilde a,\tilde c|~.
\end{align}
We define
\begin{align}
{\rho _{\bar AC}} - \frac{I}{{{d_A}}} \otimes {\rho _C} = \Delta {\rho _{\bar AC}}~.
\end{align}
Then we obtain
\begin{align}
&{\rm{Tr}}\left( {\Delta \rho _{\bar AC}^2} \right) = {\rm{Tr}}\left( {\rho _{\bar AC}^2} \right) - 2{\rm{Tr}}\left( {{\rho _{\bar AC}}\left( {\frac{I}{{{d_A}}} \otimes {\rho _C}} \right)} \right) + {\rm{Tr}}\left( {{{\left( {\frac{I}{{{d_A}}} \otimes {\rho _C}} \right)}^2}} \right)\nonumber\\
&= {\rm{Tr}}\left( {\rho _{\bar AC}^2} \right) - \frac{1}{{{d_A}}}{\rm{Tr}}\left( {\rho _C^2} \right)~.
\end{align}
The Haar integral formula allows us to estimate the trace distance as
\begin{align}
&{\rm{Tr}}\left( {\rho _{\bar AC}^2} \right) = \frac{1}{{d_A^2d_B^2}}U_{cd}^{ab}U_{\tilde ab}^{\dag \tilde cd}U_{\tilde c\tilde d}^{\tilde a\tilde b}U_{a\tilde b}^{\dag c\tilde d}~,\nonumber\\
&{\rm{Tr}}\left( {\rho _C^2} \right) = \frac{1}{{d_A^2d_B^2}}U_{cd}^{ab}U_{ab}^{\dag \tilde cd}U_{\tilde cd'}^{a'b'}U_{a'b'}^{\dag cd'}~,\nonumber\\
&\int {dU} {\rm{Tr}}\left( {\rho _{\bar AC}^2} \right) \approx \frac{1}{{{d_A}{d_C}}} + \frac{1}{{{d_B}{d_D}}}~,\nonumber\\
&\frac{1}{{{d_A}}}\int {dU} {\rm{Tr}}\left( {\rho _C^2} \right) \approx \frac{1}{{{d_A}{d_C}}} + \frac{1}{{d_A^2{d_B}{d_D}}}~,
\end{align}
where we drop the higher-order terms in the 2-design formula. Note that the leading terms in the last two formulas are the same. Thus we arrive at
\begin{align}
\int {dU} \left\| {\Delta {\rho _{\bar AC}}} \right\|_1^2 \le {d_A}{d_C}\int {dU} \left\| {\Delta {\rho _{\bar AC}}} \right\|_2^2 \le \frac{{{d_A}{d_C}}}{{{d_B}{d_D}}}~.
\end{align}
This formula says that for a small input state $d_D \gg d_A$, the total state is quickly decoupled. We could also use the 2-Renyi relative entropy and the 2-Renyi conditional mutual information
\begin{align}
&{I^{(2)}}(\bar A:C|\bar BD) = S_{\bar A\bar BD}^{(2)} + S_{\bar BCD}^{(2)} - S_{\bar BD}^{(2)}\nonumber\\
&= S_C^{(2)} + S_{\bar A}^{(2)} - S_{\bar AC}^{(2)}\nonumber\\
&= {\log _2}{\mathop{\rm Tr}\nolimits} \left( {\rho _{\bar AC}^2} \right) - {\log _2}{\mathop{\rm Tr}\nolimits} \left( {\rho _{\bar A}^2} \right) - {\log _2}{\mathop{\rm Tr}\nolimits} \left( {\rho _C^2} \right)\nonumber\\
&\approx {\log _2}\frac{{d_A^3{d_B} + {d_A}{d_B}d_D^2}}{{{d_A}{d_B} + {d_A}{d_B}d_D^2}}~.
\end{align}
If the 2-Renyi conditional mutual information is small enough, namely,
\begin{align}
d_D \gg d_A~,
\end{align}
we will arrive at an efficient decoupling.

The requirement where the conditional mutual information is small, implies that we have a high fidelity of recovery based on the Theorem \ref{recover}. Namely, one can reconstruct the Bell pair between $A$ and $\bar{A}$ with high fidelity. Beyond the Petz map, one of the explicit decoding algorithm is recently constructed by Yoshida and Kitaev \cite{Yoshida:2017non}. Thus we see again that the possibility of quantum error correction in this process, is based on the scrambling property of the Haar random unitary. 

\subsection{Error correction and U(1)-symmetric Haar randomness}
Now we could address the extension of the story about error correction and the U(1)-symmetric Haar randomness.

We could do similar heuristic calculations as we have done before. We consider the states before encoding, $\left| {{\alpha _a},B} \right\rangle $, to be charge eigenstates. We fix the total charge $m$ to be $m_A+m_B$, where $m_A \le n_A= \log_2 d_A$ is the fixed charge for $A$, and similar definitions apply for $B$. Now, the dimension of the code subspace is 
\begin{align}
\tilde{d}_A= \left( \begin{array}{l}{n_A}\\m_A\end{array} \right)~.
\end{align}
We consider the Knill-Laflamme condition
\begin{align}
\int_{{ \oplus _p}{{\cal H}_p}} d U\left\langle {{\beta _a}|O|{\beta _b}} \right\rangle  = \int_{{ \oplus _p}{{\cal H}_p}} d UU_i^{\dag a}U_b^jO_j^i~.
\end{align}
Since $a$ and $b$ have the same charge, we have
\begin{align}
\int_{{ \oplus _p}{{\cal H}_p}} d U\left\langle {{\beta _a}|O|{\beta _b}} \right\rangle  = \frac{1}{{{d_m}}}\delta _b^a\delta _i^jO_j^i = \delta _b^a{\langle {O_m}\rangle _m}~.
\end{align}
Moreover, the variance is given by
\begin{align}
&\int_{{ \oplus _p}{{\cal H}_p}} d U{\left| {\left\langle {{\beta _a}|O|{\beta _b}} \right\rangle } \right|^2} - {\left| {\int_{{ \oplus _p}{{\cal H}_p}} d U\left\langle {{\beta _a}|O|{\beta _b}} \right\rangle } \right|^2}\nonumber\\
&=  - \frac{{|{{\langle {O_m}\rangle }_m}{|^2}}}{{{d_m} + 1}}\delta _b^a + \frac{{{{\left\langle {{O_m}O_m^\dag } \right\rangle }_m}}}{{{d_m} + 1}}~.
\end{align}
The above calculations imply a similar conclusion with the Haar randomness: we will have good code property for large $d_m$. 
\subsection{U(1)-symmetric Hayden-Preskill experiment: hint for weak gravity?}
We could perform similar computations in the U(1)-symmetric Hayden-Preskill experiment. One consider the input state in the Hayden-Preskill experiment, $A$ and $B$ (with number of qubits $n_A$ and $n_B$, Hilbert space dimensions $d_A$ and $d_B$) to have fixed charge $m_A$ and $m_B$. After a U(1)-symmetric evolution, we get systems $C$ and $D$ with Hilbert space dimensions $d_C$ and $d_D$. We assume total charge $m=m_A+m_B$. 

In this setup, since we set the explicit charge for the system $A$ and $B$, the actual dimensions we should consider in the system $A$ and $B$ should be the dimensions of their charge eigenspaces. For convenience, we denote 
\begin{align}
\tilde{d}_A= \left( \begin{array}{l}{n_A}\\m_A\end{array} \right)~,~~~~~\tilde{d}_B= \left( \begin{array}{l}{n_B}\\m_B\end{array} \right)~.
\end{align}
Now, combining with the tools we have before, one could show that in the current setup
\begin{align}
&\int d U{\mathop{\rm Tr}\nolimits} \left( {\rho _{\bar AC}^2} \right) \approx \frac{{G\left( {{n_C},{n_D},m} \right)}}{{d_q^2{{\tilde d}_A}}} + \frac{{G\left( {{n_D},{n_C},m} \right)}}{{d_q^2{{\tilde d}_B}}}~,\nonumber\\
&\frac{1}{{{{\tilde d}_A}}}\int d U{\mathop{\rm Tr}\nolimits} \left( {\rho _C^2} \right) \approx \frac{{G\left( {{n_C},{n_D},m} \right)}}{{d_q^2{{\tilde d}_A}}} + \frac{{G\left( {{n_D},{n_C},m} \right)}}{{d_q^2\tilde d_A^2{{\tilde d}_B}}}~.
\end{align}
Part of the above calculations has been done in \cite{Yoshida:2018ybz}, which is in the special case where $n_C \ge m \ge n_D$. However, in our case, we still wish to assume $n_C\ge n_D$. But we wish to keep the charge $m$ to be more general, in order to inspire discussions about quantum gravity conjectures. 

Now we discuss the decoupling condition in this system. One could compute the conditional mutual information
\begin{align}
&{I^{(2)}}(\bar A:C|\bar BD) = {\log _2}{\mathop{\rm Tr}\nolimits} \left( {\rho _{\bar AC}^2} \right) - {\log _2}{\mathop{\rm Tr}\nolimits} \left( {\rho _{\bar A}^2} \right) - {\log _2}{\mathop{\rm Tr}\nolimits} \left( {\rho _C^2} \right)\nonumber\\
&\approx  - {\log _2}\left( {1 - \frac{{{{\tilde d}_A}G\left( {{n_D},{n_C},m} \right)}}{{{{\tilde d}_B}G\left( {{n_C},{n_D},m} \right)}}} \right)~.
\end{align}
So the decoupling condition is
\begin{align}
\frac{{{{\tilde d}_A}G\left( {{n_D},{n_C},m} \right)}}{{{{\tilde d}_B}G\left( {{n_C},{n_D},m} \right)}} \ll 1~.
\end{align}
Namely, we have,
\begin{align}
\frac{{G\left( {{n_C},{n_D},m} \right)}}{{{{\tilde d}_A}}} \gg \frac{{G\left( {{n_D},{n_C},m} \right)}}{{{{\tilde d}_B}}}~.
\end{align}
The decoupling condition in the case $n_C \ge m \ge n_D$, is the same as the criterion given in \cite{Yoshida:2018ybz}. We will follow the same path but instead make a more general analysis.
\begin{itemize}
\item In the case ${n_C} \ge m \ge {n_D}$, we have
\begin{align}
&G({n_C},{n_D},m) = \sum\limits_{f = m - {n_D}}^m {d_f^{\left( {{n_C}} \right)}} {\left( {d_{m - f}^{\left( {{n_D}} \right)}} \right)^2} = \sum\limits_{f = 0}^{{n_D}} {d_{m - f}^{\left( {{n_C}} \right)}} {\left( {d_f^{\left( {{n_D}} \right)}} \right)^2}~,\nonumber\\
&G({n_D},{n_C},m) = \sum\limits_{f = 0}^{{n_D}} {{{\left( {d_{m - f}^{\left( {{n_C}} \right)}} \right)}^2}d_f^{\left( {{n_D}} \right)}}~.
\end{align}
The approach of \cite{Yoshida:2018ybz} is to show that for each $f$ in the above sum, we wish to demand 
\begin{align}
\frac{{d_{m - f}^{\left( {{n_C}} \right)}{{\left( {d_f^{\left( {{n_D}} \right)}} \right)}^2}}}{{{{\tilde d}_A}}} \gg \frac{{{{\left( {d_{m - f}^{\left( {{n_C}} \right)}} \right)}^2}d_f^{\left( {{n_D}} \right)}}}{{{{\tilde d}_B}}} \Rightarrow \frac{{d_f^{\left( {{n_D}} \right)}}}{{{{\tilde d}_A}}} \gg \frac{{d_{m - f}^{\left( {{n_C}} \right)}}}{{{{\tilde d}_B}}}~.
\end{align}
In the case where $f\ge m_A$, the right-hand side is smaller than 1. For those $f$s, we demand
\begin{align}
d_f^{\left( {{n_D}} \right)} \gg {{\tilde d}_A} \Rightarrow \left( {\begin{array}{*{20}{c}}
{{n_D}}\\
f
\end{array}} \right) \gg \left( {\begin{array}{*{20}{c}}
{{n_A}}\\
{{m_A}}
\end{array}} \right)~,
\end{align}
which is correct when $n_D\gg  n_A$. Although for small $f$, the above inequality may not hold, one could argue that there are relatively small contributions from those $f$s, since for large $n_C$, it forms nearly a binomial distribution and the contribution $f\sim \mathcal{O}(n_D)$ should be dominant. Thus, the conclusion is that, in this case, $n_D\gg n_A$ will generically ensure the decoupling. 
\item We have noticed that the above computation could completely extend to the case where we have $m \ge {n_C} \ge {n_D}$, since 
\begin{align}
&G({n_C},{n_D},m) = \sum\limits_{f = m - {n_D}}^{{n_C}} {d_f^{\left( {{n_C}} \right)}} {\left( {d_{m - f}^{\left( {{n_D}} \right)}} \right)^2}~,\nonumber\\
&G({n_D},{n_C},m) = \sum\limits_{f = m - {n_C}}^{{n_D}} {d_f^{\left( {{n_D}} \right)}} {\left( {d_{m - f}^{\left( {{n_C}} \right)}} \right)^2} = \sum\limits_{f = m - {n_D}}^{{n_C}} {d_{m - f}^{\left( {{n_D}} \right)}} {\left( {d_f^{\left( {{n_C}} \right)}} \right)^2}~.
\end{align}
Thus, we will still win to show the inequality in each term and the $f\sim \mathcal{O}(n_D)$ dominance in the sum.
\item The problem appears in the case of small charge, ${n_C} \ge {n_D} \ge m$. We have
\begin{align}
&G({n_C},{n_D},m) = \sum\limits_{f = 0}^m {d_f^{\left( {{n_C}} \right)}} {\left( {d_{m - f}^{\left( {{n_D}} \right)}} \right)^2}~,\nonumber\\
&G({n_D},{n_C},m) = \sum\limits_{f = 0}^m {d_f^{\left( {{n_D}} \right)}} {\left( {d_{m - f}^{\left( {{n_C}} \right)}} \right)^2}~.
\end{align}
Since the sum over $f$ is restricted to the range between 0 and $m$, it has to be dominated by small $f$ pieces if $m$ is small enough. In this case, the condition $n_D\gg n_A$ may not ensure the decoupling. We will give a simple example. Consider $m=2$ and $m_A=m_B=1$. We have
\begin{align}
&\frac{{G({n_C},{n_D},m)}}{{{{\tilde d}_A}}} = \frac{1}{{{n_A}}}\left( {\frac{{n_D^2{{({n_D} - 1)}^2}}}{4} + {n_C}n_D^2 + \frac{{{n_C}({n_C} - 1)}}{2}} \right)~,\nonumber\\
&\frac{{G({n_D},{n_C},m)}}{{{{\tilde d}_B}}} = \frac{1}{{{n_B}}}\left( {\frac{{n_C^2{{({n_C} - 1)}^2}}}{4} + {n_D}n_C^2 + \frac{{{n_D}({n_D} - 1)}}{2}} \right)~.
\end{align}
For large system, asymptotically we have,
\begin{align}
&\frac{{G\left( {{n_C},{n_D},m} \right)}}{{{{\tilde d}_A}}} \sim \frac{{n_D^4}}{{{n_A}}}~,\nonumber\\
&\frac{{G\left( {{n_D},{n_C},m} \right)}}{{{{\tilde d}_B}}} \sim \frac{{n_C^4}}{{{n_B}}}~.
\end{align}
Thus, as long as
\begin{align}
\frac{{n_D^4}}{{{n_A}}} \le \frac{{n_C^4}}{{{n_B}}}~,
\end{align}
we cannot arrive at the previous decoupling even if $n_D\gg n_A$. 
\end{itemize}
The decoupling condition beyond the usual $n_D \gg n_A$ criterion, in the case of the small charge, is seemingly intuitive physically. If the system has a fixed small amount of charge, the code subspace has stronger restrictions. Thus, the decaying process is hard to proceed to arrive at the approximate Markov condition. 

At this point, we think it is safe to make the following remark,
\begin{remark}\label{wgc}
Assuming the decoupling condition $n_D \gg n_A$ generically, there exists a lower bound on the charge $m$.
\end{remark} 
Moreover, if we assume the dominance of binomial distribution, we could arrive at the condition when it is safe to make $m\sim \mathcal{O}(n_D)$. If so, the previous proof for larger charges could still apply, and we arrive at a safe decoupling at $n_D \gg n_A$. We wish to interpret the above remark as a hint of weak gravity conjecture, although it is still far from an explicit statement about the real life of quantum gravity and black holes. We will make further comments at the end of this section.

\subsection{Bounds from the approximate Eastin-Knill theorem}
Here we discuss the connection between this work and the recently-proved approximate Eastin-Knill theorem \cite{Faist:2019ahr}. 

As indicated from previous discussions, symmetry could provide a significant constraint on the structure of quantum error correction codes. When the symmetry operators \emph{commute} with the code subspace, we say that the code is \emph{covariant}. To be more precise, we consider the encoding from the logical system $S_\text{logical}$ to the physical system $S_\text{physical}$. The physical system is made by $D$ subsystems $\Lambda_1,\cdots,\Lambda_D$. We assume that there is a charge operator $T_\text{logical}$ acting on the logical system, while the charge operator $T_{\text{physical}}$ acting on the physical system. Furthermore, we assume that the symmetry acts on the physical system transversely. Namely, we have
\begin{align}
{T_{{\rm{physical}}}} = \sum\limits_{i = 1}^D {{T_i}} ~,
\end{align}
and the condition of covariance requires that the symmetry operator should commute with the encoding.

When taking a look at the structure of the U(1)-symmetric Hayden-Preskill experiment, we could immediately realize that it is an approximate covariant code with U(1) symmetry. The symmetry operator is the charge
\begin{align}
T_\text{physical}=Q=\sum_{i=1}^{D} \frac{1+Z_{i}}{2}~,
\end{align}
which is definitely transverse. When we split the system by $A$ and $B$, we restrict their charges to be $m_A$ and $m_B$. The encoding map commutes with the symmetry operator, ensuring the covariance of the code.

A celebrated result is obtained by Eastin and Knill \cite{EK}, claiming that if the symmetry is continuous, then the exact covariant code does not exist. Recently, there are detailed discussions about extending this idea to approximate quantum error correction \cite{Faist:2019ahr}. There are universal constraints written in terms of fidelity from properties of symmetries. 

\begin{thm}[Approximate Eastin-Knill theorem]
Say that the symmetry is continuous. Define the worst-case entanglement fidelity,
\begin{align}
&f_{{\rm{worst}}}^2 ={\max _{\cal R}} {\min _{|\phi \rangle  \in \mathcal{C}}}\langle \phi |\mathcal{R}(|\phi \rangle \langle \phi |)|\phi \rangle ~,\nonumber\\
&\epsilon_{{\rm{worst}}}^2=1-f_{{\rm{worst}}}^2 ~,
\end{align}
for given noise and encoding, where $\mathcal{C}$ is the code subspace. Then we have
\begin{align}
{\epsilon _{{\rm{worst }}}} \ge \frac{{\Delta {T_{{\rm{logical}}}}}}{{2D{{\max }_i}\Delta {T_i}}}~,
\end{align}
where ${\Delta {T_{{\rm{logical}}}}}$ represents the difference between maximal and minimal values of charge operator eigenvalues in the logical system. $\Delta T_i$ is the difference between maximal and minimal values of charge operator eigenvalues for system $i$.
\end{thm}

Since the theorem has already been proven, it implies that in the U(1)-symmetric Hayden-Preskill experiment, there should exist the above universal constraint even if decoupling happens. Thus, we have the following statement,
\begin{remark}
The approximate Eastin-Knill theorem set bounds on the fidelity of the U(1)-symmetric Hayden-Preskill experiment even when decoupling happens.
\end{remark}

A direct computation about the above fact goes as the following. One can show that \cite{Faist:2019ahr}, the fidelity defined by the inner product, is the same as the fidelity defined by the trace distance on the density matrix. Furthermore, the worst-case fidelity searches for the worst performance in the code subspace, while in the U(1)-symmetric Hayden-Preskill, we are using a specific state. Furthermore, we know from Theorem \ref{recover} that conditional mutual information is bounded by the fidelity. Combining all the above facts together, we demand the consistency condition between the decoupling of the U(1)-symmetric Hayden-Preskill and the approximate Eastin-Knill theorem:
\begin{align}
f_{{\rm{worst }}}^2 \le 1 - \frac{{m_A^2}}{{4{D^2}}} \Rightarrow {\log _2}\left( {1 - \frac{{m_A^2}}{{4{D^2}}}} \right) \gtrapprox {\log _2}{\mathop{\rm Tr}\nolimits} \left( {\rho _{\bar A}^2} \right) + {\log _2}{\mathop{\rm Tr}\nolimits} \left( {\rho _C^2} \right) - {\log _2}{\mathop{\rm Tr}\nolimits} \left( {\rho _{\bar AC}^2} \right)~.
\end{align}
The bound is following from the fact that, in such a system, we have $\Delta T_\text{logical}=m_A$ and $\Delta T_i=1$. When decoupling happens, we have
\begin{align}
1 - \frac{{m_A^2}}{{4{D^2}}} \gtrapprox 1 - \frac{{\left( {\tilde d_A^2 - 1} \right)G\left( {{n_D},{n_C},m} \right)}}{{{{\tilde d}_A}{{\tilde d}_B}G\left( {{n_C},{n_D},m} \right)}}~,
\end{align}
which is, approximately
\begin{align}
\frac{{m_A^2}}{{4{D^2}}} \lessapprox \frac{{{{\tilde d}_A}G\left( {{n_D},{n_C},m} \right)}}{{{{\tilde d}_B}G\left( {{n_C},{n_D},m} \right)}}~.
\end{align}
Although we do not have a rigorous verification for the above combinatorial inequality, some heuristic interpretations might be given. The right-hand side is seeking for a bound on fidelity by looking at permutations using the rule of the Haar randomness and performing some counting. The left-hand side is looking at a much weaker constraint by counting the number of qubits in the code subspace and the whole space directly. Thus heuristically, the specific construction based on the Haar randomness should be weaker than the universal constraints. In the limit where $n_D \gg n_A$ and sufficiently large charge, where we expect decoupling should happen, both sides vanish. Thus, the above inequality will provide constraints on the limitation of decouplings that might happen in the U(1)-charged Hayden-Preskill experiment. 

\subsection{Comments on quantum gravity conjectures}
In this part, we wish to present the study of the relationship between quantum information computations on the Hayden-Preskill experiment with the related quantum gravity conjectures. More precisely, we will discuss the existence and properties of global and gauge symmetries in quantum gravity. 

The above computations are far from real quantum gravitational systems like black holes. For instance, since we directly put a random unitary to represent black hole dynamics, there is no precise definition of the locality. Thus, it is hard to distinguish global and gauge symmetries, allowing us the opportunity to discuss both of them at the same time. However, we are still expecting that we could get some hints from toy models, following the spirit of several important works about toy models in quantum gravity (for instance, \cite{Hayden:2007cs}). 

People have long-time suspicion in the past that there is no precise notion of global symmetry in theories with consistent quantum gravity. An important argument is from the black hole no-hair theorem, stating that black holes are parametrized only by their mass, charge, and angular momentum. Since this charge means the local charge corresponding to gauge symmetries, there is no room for global charges to existing. Namely, suppose a precise notion of global symmetry exists, there is a huge amount of states which our descriptions about black holes cannot capture, causing large degeneracies and a thus large amount of entropy. If we don't wish it to happen, we should forbid the notion of global symmetry in the quantum gravitational system.  

Recently, people make significant progress about global symmetries in quantum gravity, based on holography and quantum information theory. \cite{Harlow:2018jwu,Harlow:2018tng} presents a holographic physical proof of no-global-symmetry conjecture based on the following two steps: First, they use the tools from holography, so-called entanglement wedge reconstruction, to argue that bulk global symmetries should act on the boundary transversally in each subregion. Second, they argue that these symmetries acting on the boundaries are the logical operators that preserve the code subspace. From quantum information theory, they conclude that those symmetries must be the logical identity, justifying the no-global-symmetry conjecture in the bulk. 

The proof made by \cite{Harlow:2018jwu,Harlow:2018tng}, especially the first step, is closely related to the proof of Eastin-Knill theorem stating that there is no covariant code for continuous symmetries. Moreover, one could also construct a holographic proof of the Eastin-Knill theorem (see \cite{Harlowtalk,Faist:2019ahr}). This follows from the fact that the holographic dictionary could be interpreted as a quantum error correction code \cite{Dong:2016eik,Almheiri:2014lwa}, while covariant code means that the global symmetry exists in the bulk effective field theory, namely, the code subspace of the boundary Hilbert space.

Here, in the context of the Hayden-Preskill experiment, the existence of the (approximate) Eastin-Knill theorem has another interpretation: it could set a lower bound on the fidelity of recovery during charged black hole evaporates. The structure of charges provides crucial constraints on the code subspace, forbid the black hole system to decouple sufficiently, and arrive at a chaotic, Markov state. 

Another side of the story comes from the claim about the (weak-form) weak gravity conjecture. While the no-global-symmetry theorem suggests a fundamental property about global symmetry in the bulk spacetime, weak gravity conjecture is a claim about gauge symmetry. The slogan of the conjecture is very simple: gravity seems to be the weakest force in any theory that is consistent. More precisely, consider a quantum gravity theory associated with gauge symmetries, there always exists a state such that the charge-to-mass ratio is larger than some constants. Considering the existence of the Reissner-Nordstrom metric in the semi-classical description of gravity, we could set this constant to be $1/M_\text{Planck}$ in the natural unit. 

We wish to argue that the existence of a lower bound on the charge in order to make the decoupling condition $d_D\gg d_A$ universal might be related to some heuristic arguments of weak gravity. In fact, it is pretty consistent with one of the earliest arguments that are presented in \cite{ArkaniHamed:2006dz}: suppose that weak gravity conjecture is false, then there are many states that have relatively small charges. Considering that during the black hole decay process, the charge is conserved while mass is decreasing, then there exists at least one decaying product with a growing charge-to-mass ratio. Then we arrive at the conclusion that there are many stable states which completely cannot decay, which is not favored for a consistent quantum gravitational theory. This situation is pretty similar to what we have in the U(1)-symmetric Hayden-Preskill experiment. If the charge we have is sufficiently small enough, black hole quantum states are hard to decouple even in the case where $d_D\gg d_A$, since small charge restrict the allowed code subspace. If we believe that the black hole decay process will end up with a highly chaotic state with a sufficiently large amount of decoupling, we have to make the initial state sufficiently charged. Thus the U(1)-symmetric Hayden-Preskill experiment could be regarded as a toy example about the decay of a charged black hole, supporting the statement of weak gravity conjecture. Furthermore, if we believe that this experiment happens in the 1+1 dimensions, we might regard the size of the qubit system as the mass. Thus the condition $m \ge \mathcal{O}\left(n_{D}\right)$ suggests a bound relating charge and mass, sharpening the analogy between this experiment and weak gravity conjecture.

Finally, we briefly address issues about gauge/global symmetries and holography. Here our Hayden-Preskill experiment is at least manifestly, built in the bulk. We are not able to directly build gauge groups using random unitaries in this work. However, in the discussions about chaotic systems, we would imagine that the decoupling process could appear in the boundary. People believe that roughly speaking, gauge symmetries in the bulk are dual to global symmetries in the boundary. Thus we could, in principle, using constructions like complex SYK model to probe gauge symmetries in the bulk. We look forward to future studies about concrete models with gauge/global symmetry duality in the future.

\section{Discrete symmetries}\label{DS}
Although the main part of this paper is focusing on continuous symmetry, we wish to address a closely related issue, discrete symmetry. We claim that unlike continuous symmetries, discrete symmetries are more similar to systems without any symmetry: many properties are not far from the full Haar randomness. 

The distribution of the Haar randomness on the unitary group could be generalized easily to the unitary group with extra discrete constraints. For a $D$-qubit system, extra constraints induced by symmetry $\mathcal{S}$ will end up with a uniform distribution on the quotient space: $\text{U}(2^D)/\mathcal{S}$. One of the earliest works about symmetries in random systems is due to Dyson \cite{dys} about the three-fold classification of random matrix theory. In his classification, we take $\mathcal{S}=\text{O}(2^D)$, inducing the circular orthogonal ensemble (COE), or we could also take $\mathcal{S}=\text{Sp}(2\times 2^D)$, inducing the circular symplectic ensemble (CSE). these two extended classes could be understood as random systems with time-reversal invariance $T$. COE corresponds to $T^2=1$ while CSE corresponds to $T^2=-1$. A more detailed classification, so-called Altland-Zirnbauer \cite{Altland:1997zz} classification, involves ten classes in total involving at most two antiunitary operators, is widely used in the study of condensed matter physics, for instance, ten-fold classification of topological insulators \cite{Ryu:2010zza,topokitaev}.

Unlike continuous symmetries, discrete symmetries often reflect topological properties in the many-body system. In random systems in general, some chaotic properties, for instance, level distribution or spectral form factor, will show universal behaviors associated with corresponding discrete symmetries. A three(ten)-fold of random systems could naturally induce a three(ten)-fold classification of quantum circuits. One could generalize those Weingarten functions easily through group theory. Some earlier works have been done in \cite{random,CS}, and those topics are recently summarized in \cite{thesis,Hunter-Jones:2018otn,Cotler:2019egt,Liu:2018hlr} in terms of OTOCs, random quantum circuits, and chaotic systems in general. 

In the context of discrete symmetries, Weingarten calculations are very similar to the Haar average without symmetry. For instance, we list some Weingarten functions for COE for large $L$:
\begin{align}
&{\rm{Wg}}(1) = \frac{1}{{L + 1}}~,\nonumber\\
&{\rm{Wg}}(1,1) = \frac{{L + 2}}{{L(L + 1)(L + 3)}}~,\nonumber\\
&{\rm{Wg}}(2) =  - \frac{1}{{L(L + 1)(L + 3)}}~.
\end{align}
Comparing with the results with the whole random unitaries, we see that asymptotically, extra symmetries cannot affect the scaling on $L$. In fact, one could check that it is indeed a generic feature of discrete symmetries. With the help of earlier works (for instance \cite{random}), we could easily check the following claim

\begin{remark}
Discrete topological symmetries cannot change the decoupling properties of random unitaries in the Dyson and Altland-Zirnbauer classifications. More precisely, the decoupling fidelities have the same asymptotic scaling as the random unitary of the whole unitary group, both in the contexts of Page theorem and Hayden-Preskill experiment in the large system. 
\end{remark}

The above generic feature could also be reflected in the language of quantum error correction, showing the striking differences between continuous and discrete symmetries. Intuitively, from the scrambling point of view, since Weingarten functions have the same scaling behavior, there are no extra limitations for a code arriving at the end of decoupling. In fact, unlike the Eastin-Knill theorem, discrete symmetries are allowed in quantum information theory, in the context of covariant codes.

\begin{thm}
There exist covariant codes with discrete symmetries. 
\end{thm}

In fact, from the earliest understanding about quantum error correction in AdS/CFT, the qutrit code discussed by \cite{Dong:2016eik}, is an example of covariant code. In fact, there exist logical operators $X_\text{logical}$ and $Z_\text{logical}$ in the code subspace that could be viewed as global discrete symmetries. Some other examples and proofs are summarized in \cite{Hayden:2017jjm}. 

Note that such a situation is different from quantum gravity and holography. People believe that the no-global-symmetry conjecture is correct even for discrete symmetries in a consistent theory of quantum gravity. Thus, toy models constructed from covariant codes with discrete symmetries cannot capture features about global symmetries for AdS/CFT codes. So how it works from holographic statements such that it proves a stronger version of the Eastin-Knill theorem? Here, we wish to take the point of view from \cite{Faist:2019ahr}. In the two-step proof by Harlow and Ooguri \cite{Harlow:2018jwu,Harlow:2018tng}, the second step, where symmetries are logical operators acting on each boundary, may not be preserved for generic codes that are used for fault-tolerant quantum computation. This is from the fact that we expect in the boundary quantum field theory, the code subspace is from the low energy spectrum of the corresponding CFT. Roughly speaking, symmetry operators acting on the boundary are required not to change the energy of the states too much, justifying that they are logical operators. 

In the recent discussions about quantum gravity, people make extensive studies on the Jackiw-Teitelboim gravity in two dimensions, as the dual gravitational theory of the SYK model. There are recently many discussions about discrete symmetries, and also continuous symmetries in such holographic context (see \cite{Liu:2019niv,Sachdev:2019bjn,Stanford:2019vob,Kapec:2019ecr,Iliesiu:2020qvm}). It is interesting and important to explore symmetries in those concrete models, and their scrambling and quantum error correction properties during quantum information processing. 

\section{Outlook}\label{CON}
In this section, we will point out some potential applications and possible future research directions along the line of this work.

\subsection{Gauge and global symmetries}
As we discuss in the main text, we could either build the Hayden-Preskill experiment directly with symmetries, or our stories about scrambling happen in the boundary. For the former, it will be interesting, in the future, to build concrete models about the Hayden-Preskill experiment, or in general, quantum error correction codes with gauge symmetry. In the context of AdS/CFT code, considerable progress has been made along the line of \cite{Pastawski:2015qua,Harlow:2015lma,Donnelly:2016qqt}. For the latter, it might be interesting to consider the holographic Hayden-Preskill experiment following the setup of \cite{Almheiri:2019psf} with global symmetry in the boundary. The problem setup discussed in the main text is far from reality. Thus, more realistic models, especially with a concrete definition of gauge symmetries, will definitely be helpful for the understanding of black hole evaporation associated with symmetries in its low energy description. Moreover, the implementation of gauge symmetries may not only be helpful for theoretical studies of tensor networks, quantum gravity, and AdS/CFT, but also for machine learning and neural networks. For instance, see \cite{learning,learning2}.

\subsection{Black hole thought experiments from toy models}
It will be helpful for considering more details in the U(1)-symmetric Hayden-Preskill experiment discussed in Section \ref{CO}. For instance, what is the precise sense of the approximate Eastin-Knill theorem here in the high energy physics sense, without using AdS/CFT? How to connect the discussions about lower bound on the charge, to the evaporation experiments in the statement of weak gravity conjecture? Is it possible to address the entropic arguments used in \cite{Cheung:2018cwt,Cheung:2019cwi} in the context of U(1)-symmetric Hayden-Preskill experiment? The last suggestion might require a generalization from Bell pair to thermofield double in the Hayden-Preskill-type experiment. It might also be interesting to discuss the connection between quantum gravity conjectures and the recovery construction, beyond the Petz map or Yoshida-Kitaev decoding \cite{Yoshida:2018ybz}, moreover, in the context of traversable wormholes \cite{Gao:2016bin,Maldacena:2017axo}. 

Recently, significant progress has been made along the line of the Hayden-Preskill experiment, black hole evaporation, and the information paradox. With the help of tools from the gravitational path integral, holography, and entanglement, people reproduce the Page curve, a standard result of black hole unitarity from random unitary toy models \cite{Penington:2019npb,Almheiri:2019psf,Penington:2019kki,Almheiri:2019qdq,Marolf:2020xie}. It might be interesting to address those calculations about black hole evaporation with global or gauge symmetries into account, in generic semiclassical descriptions following \cite{Penington:2019npb} or in explicit models like Jackiw-Teitelboim gravity or the SYK model. Addressing those calculations in a holographic context will be helpful for interpretations of symmetries in quantum gravity and quantum information processing in black hole dynamics.

\subsection{Experimental platforms and large-scale bootstrap}
Nowadays, there are rising interests from formal constructions of black hole thought experiment to quantum simulation in the analog or digital, real platform in code-atomic and condensed matter physics \cite{Brown:2019hmk}. These studies might be helpful for the study of quantum information science itself, providing clear targets and motivations for benchmarking cold-atomic devices or near-term quantum devices. Moreover, it might also be helpful towards a deeper understanding of black hole dynamics and holography in the near-term or long-term, and clarifying conceptual problems about computability and stimulability of our universe, namely, the quantum Church-Turing Thesis \cite{Bouland:2019pvu,Susskind:2020kti,Kim:2020cds}.

Charged systems in quantum many-body physics could provide fruitful platforms for exploring the propagation of quantum information in the quantum materials. For instance, significant progress has been made about thermoelectric transport in the complex SYK theoretically and experimentally \cite{Kruchkov:2019idx}. Here, we wish to emphasize recent developments of conformal bootstrap for exactly solving CFTs with global symmetry. In the context of the four-point function, global symmetries will decompose crossing equations in different sectors based on their representation theory, causing generically different critical exponents for corresponding CFTs. With novel analytic and numerical tools developed recently \cite{Poland:2018epd}, people could generate information about global symmetries automatically \cite{Go:2019lke} for a large class of Lie group, and solve bootstrap equations in a relatively large scale \cite{Landry:2019qug,Chester:2019ifh}. A particularly successful example recently is O(2) symmetric CFT in three dimensions \cite{Chester:2019ifh} (very similar to U(1)), which is widely applied in condensed matter and cold-atomic physics. Those studies about CFTs will potentially be very helpful for providing robust data for AdS/CFT and theoretical/experimental studies of quantum materials. 

\subsection{Quantum simulation}
Here we say a few words about quantum simulation for charged systems. It is definitely interesting to construct further toy models in quantum circuits to simulate some specific charged black hole evaporation processes either in holography or in the bulk. A particularly interesting example is \cite{Horowitz:2016ezu}, where arguments about weak gravity conjecture could help prevent the appearance of the naked singularity. It might be interesting to address possible quantum information interpretations or check what happens for such processes in the context of holography.

Another interesting connection is about quantum circuits and hydrodynamics. Opposite to discussions mostly in the main text, people often treat U(1)-symmetric random circuits as a specific random circuit model for energy conservation. One could assign possible macroscopic hydrodynamical variables to quantum circuits, to study emergent phenomena of the microscopic Haar randomness. This study could be helpful for understandings of quantum circuits themselves (for instance, the efficiency of approaching $k$-designs \cite{Nahum:2017yvy,un}, or emergent phenomena like replica wormholes in open quantum systems \cite{Piroli:2020dlx}), but also be helpful for hydrodynamics itself. Can we simulate classical, emergent, novel hydrodynamical effects in quantum circuits? What is the meaning of concepts, like non-Newtonian fluid or Reynolds number in quantum circuits? Thanks to holography, it is not completely irrelevant to quantum gravity \cite{Bredberg:2011jq}. 

\subsection{Other suggestions}
Here we collect some other suggestions for future research:
\begin{itemize}
\item Finite temperature. It might be interesting to generalize calculations in this paper to finite temperature, for instance, finite temperature version of spectral form factors, frame potentials or OTOCs, or Hayden-Preskill experiment with finite temperature (thermofield double). The Boltzmann distribution makes the calculation non-trivial, but more realistic to claim a connection between toy models and finite temperature black holes.
\item More symmetries and exotic quantum matter. Here we only discuss U(1) symmetry as an example of continuous symmetries. However, it will be interesting to generalize it to more general, higher symmetries. Recently, people discuss several topics about the anomaly, global symmetries, field theory dualities, connecting the study of QCD to topological materials (for instance, see \cite{Gaiotto:2017yup}). It might be interesting to discuss the theory of quantum information in such systems with higher, more complicated symmetries (for instance, \cite{Kapec:2019ecr}). 
\end{itemize}

\section*{Acknowledgement}
I thank my advisors Cliff Cheung, David Simmons-Duffin, and John Preskill, for their numerous supports and discussions during the course of the project, especially through my hard time. I also thank John Preskill and Pengfei Zhang for their related discussions and collaborations in progress. I thank Victor Albert, Fernando Brand$\tilde{\text{a}}$o, Masanori Hanada, Patrick Hayden, Alexei Kitaev, Richard Kueng, Aitor Lewkowycz, Hirosi Ooguri, Geoffrey Penington, Xiaoliang Qi, Subir Sachdev, Burak Sahinoglu, Eva Silverstein, Douglas Stanford, Brian Swingle, Eugene Tang, Gonzalo Torroba, Tian Wang, Yi-Zhuang You and Sisi Zhou for related discussions since 2017, Mark Wilde and Beni Yoshida for helpful correspondence on the draft. I also thank Sterl Phinney for offering me the opportunity to teach hydrodynamics at Caltech (ph136b 2018-2019 winter term), inspiring my interests in charged quantum circuits. Finally, I wish to express a special acknowledgment to Robert Brandenberger for his recent concern. I am also honored to attend the MURI team discussions about quantum error correction code and quantum gravity as one of quantumists from Preskill's group at Caltech. 
\\
\\
I am supported in part by the Institute for Quantum Information and Matter (IQIM), an NSF Physics Frontiers Center (NSF Grant PHY-1125565) with support from the Gordon and Betty Moore Foundation (GBMF-2644), by the Walter Burke Institute for Theoretical Physics, and by Sandia Quantum Optimization \& Learning \& Simulation, DOE Award \#DE-NA0003525.

\appendix 
\section{Engineering the Haar integrals}\label{gra}
There are several programming packages for computing the Haar integrals systematically (for instance, \cite{package1,package2}). Here, to be self-complete, we will introduce some technical engineering details about the \texttt{RTNI} package appearing recently \cite{Fukuda:2019pzs}, which cover many technical computations about the Haar integrals in this paper.

As discussed in the main text, one could usually deal with problems about the Haar integral in the following form,
\begin{align}
&\int_\mathcal{H} {U_{{j_1}}^{{i_1}} \ldots U_{{j_p}}^{{i_p}}U_{{j_1}'}^{\dagger,{i_1}'} \ldots U_{{j_p}'}^{{\dagger,i_p}'}dU}~, \nonumber\\
&= \sum\limits_{\alpha ,\beta  \in {S_p}} {\delta _{{j_{\alpha (1)}}'}^{{i_1}} \ldots \delta _{{j_{\alpha (p)}}'}^{{i_p}}\delta _{{j_1}}^{{i_{\beta (1)}}'} \ldots \delta _{{j_p}}^{{i_{\beta (p)}}'}{\rm{Wg}}({\alpha ^{ - 1}}\beta )} ~,
\end{align}
where $\alpha,\beta$ are elements of permutation group $S_p$ over $1,2,\cdots,p$. In many cases, especially when $U$ appears many times in the integral, it is hard to sum over so many $\delta$ algebras by hand. Thus, it will be much more efficient to consider computer algebra. The software \cite{Fukuda:2019pzs} is especially useful for dealing with the following types of integral
\begin{align}
{X_1}{U_1}{X_2}{U_3} \cdots {X_n}{U_n}~,
\end{align}
or 
\begin{align}
\operatorname{Tr}\left(X_{1} U_{1} X_{2} U_{3} \cdots X_{n} U_{n}\right)~,
\end{align}
in the compact form, where
\begin{align}
{U_j} \in \left\{ {U,{U^\dag },{U^T},{U^*}} \right\}~.
\end{align}
In this paper, we will only use the case for $U^\dagger$. 

One could use the function \texttt{MultinomialexpectationvalueHaar} to obtain analytic expressions for correlators in different types. For instance, we consider computing
\begin{align}
\texttt{MultinomialexpectationvalueHaar[d, \{1, 2, 1, 2\}, \{X1, Y1, X2, Y2\}, False]}~,
\end{align}
which means 
\begin{align}
\int {dU} \left( {{X_1}U{Y_1}{U^\dag }{X_2}U{Y_2}{U^\dag }} \right)~.
\end{align}
Here, \texttt{d} is the dimension of total Hilbert space, $\texttt{\{X1, Y1, X2, Y2\}}$ means the operators appearing in the integral, $\{1, 2, 1, 2\}$ is the index list ($j$s) with the set ${U_j} \in \left\{ {U,{U^\dag },{U^T},{U^*}} \right\}$, in order, of $U$s appearing in the integral, and $\texttt{False}$ means computation without the trace (otherwise it will be \texttt{True}). The output is as expected,
\begin{align}
\frac{{(\texttt{X1.X2(dTr[Y1]Tr[Y2] - Tr[Y2.Y1])} + \texttt{X1Tr[X2]( - Tr[Y1]Tr[Y2] + dTr[Y2.Y1]))}}}{\texttt{(d( - 1 + d*d))}}~.
\end{align}
Furthermore, the package has nice graph representations provided associated with integrals, consistent with quantum circuit graphs appearing in the literature. For more details, see the introduction of the package itself in \cite{Fukuda:2019pzs} or corresponding help documents. For theoretical illustrations, there is a useful paper here, especially for computing the Haar integrals for small moments \cite{Rich}.

\section{Form factors specifying diagonal and off-diagonal terms}\label{alter}
This is a technical note on the alternative representation of spectral form factors. We could specify the type of sums in the spectral form factor by isolating diagonal and off-diagonal terms. We introduce the notation
\begin{align}
&P_1=R_1~,\nonumber\\
&P_{2}^{{{\cal E}_p}}({d_p}) = \sum\limits_{a \ne b} {\int {D\lambda } {e^{i({\lambda _{p,a}} - {\lambda _{p,b}})}}}~,\nonumber\\
&P_{21}^{{{\cal E}_p}}({d_p}) = \sum\limits_{a \ne b} {\int {D\lambda } {e^{i({\lambda _{p,a}} + {\lambda _{p,b}})}}}~,\nonumber\\
&P_{22}^{{{\cal E}_p}}({d_p}) = \sum\limits_{a \ne b} {\int {D\lambda } {e^{i({2\lambda _{p,a}} - {\lambda _{p,b}})}}} \nonumber\\
&P_3^{{{\cal E}_p}}({d_p}) = \sum\limits_{a \ne b \ne c} {\int {D\lambda } {e^{i({\lambda _{p,a}} + {\lambda _{p,b}} - {\lambda _{p,c}})}}}~,\nonumber\\
&P_{31}^{{{\cal E}_p}}({d_p}) = \sum\limits_{a \ne b\ne c} {\int {D\lambda } {e^{i({2\lambda _{p,a}} - {\lambda _{p,b}}-{\lambda _{p,c}})}}}~,\nonumber\\
&P_4^{{{\cal E}_p}}({d_p}) = \sum\limits_{a \ne b \ne c \ne d} {\int {D\lambda } {e^{i({\lambda _{p,a}} + {\lambda _{p,b}} - {\lambda _{p,c}} - {\lambda _{p,d}})}}}~.
\end{align}
The form factors could also be represented in terms of $P$s for the U(1)-symmetric system. For instance, we have 
\begin{align}
R_2^{{ \oplus _p}{\mathcal{E}_p}}(L) =L+ \sum_{p } {P_2^{{\mathcal{E}_p}}({d_p})}  + \sum_{p \ne q} {P_1^{{\mathcal{E}_p}}\left( {{d_p}} \right)P_1^{{\mathcal{E}_q}*}\left( {{d_q}} \right)}~,
\end{align}
and 
\begin{align}
&R_4^{{ \oplus _p}{{\cal E}_p}}(L) = \sum\limits_{p \ne q \ne u \ne v} {P_1^{{{\cal E}_p}}({d_p})P_1^{{{\cal E}_q}}({d_q})P_1^{{{\cal E}_u}*}({d_u})P_1^{{{\cal E}_v}*}({d_v})} \nonumber\\
&+ 4{\rm{Re}}\sum\limits_{p \ne q \ne u} {P_2^{{{\cal E}_p}}({d_p})P_1^{{{\cal E}_q}}({d_q})P_1^{*{{\cal E}_u}}({d_u})}  + 2{\rm{Re}}\sum\limits_{p \ne q \ne u} {P_{21}^{{{\cal E}_p}}({d_p})P_1^{*{{\cal E}_q}}({d_q})P_1^{*{{\cal E}_u}}({d_u})} \nonumber\\
&+ 2{\rm{Re}}\sum\limits_{p \ne q \ne u} {P_1^{{\cal E}_p^2}({d_p})P_1^{{{\cal E}_q}*}({d_q})P_1^{{{\cal E}_u}*}({d_u})} \nonumber\\
&+ 4{\rm{Re}}\sum\limits_{p \ne q} {P_3^{{{\cal E}_p}}({d_p})P_1^{*{{\cal E}_q}}({d_q})}  + 2{\rm{Re}}\sum\limits_{p \ne q} {P_2^{{{\cal E}_p}}({d_p})P_2^{{{\cal E}_q}}({d_q})} \nonumber\\
&+ {\rm{Re}}\sum\limits_{p \ne q} {P_{21}^{{{\cal E}_p}}({d_p})P_{21}^{*{{\cal E}_q}}({d_q})}  + 4{\rm{Re}}\sum\limits_{p \ne q} {P_{22}^{{{\cal E}_p}}({d_p})P_1^{{{\cal E}_q}*}({d_q})} \nonumber\\
&+ 2{\rm{Re}}\sum\limits_{p \ne q} {P_1^{{\cal E}_p^2}({d_p})P_{21}^{{{\cal E}_q}*}({d_q})}  + 4(L - 1)\sum\limits_{p \ne q} {P_1^{{{\cal E}_p}}\left( {{d_p}} \right)P_1^{{{\cal E}_q}*}\left( {{d_q}} \right)}  + \sum\limits_{p \ne q} {P_1^{{\cal E}_p^2}\left( {{d_p}} \right)P_1^{{\cal E}_q^2*}\left( {{d_q}} \right)} \nonumber\\
&+ \sum\limits_p {P_4^{{{\cal E}_p}}({d_p})}  + 2{\rm{Re}}\sum\limits_p {P_{31}^{{{\cal E}_p}}({d_p})}  + 4(L - 1)\sum\limits_p {P_2^{{{\cal E}_p}}({d_p})}  + \sum\limits_p {P_2^{{\cal E}_p^2}({d_p})} \nonumber\\
&+ 2{L^2} - L~.
\end{align}

\bibliographystyle{utphys}
\bibliography{Biblio}

\end{document}